\documentclass[11pt]{article}
\usepackage{amsthm,amssymb,amsfonts,amsmath,amscd}
\usepackage{nccmath}
\usepackage{graphicx}
\usepackage[pdfa]{hyperref}

\makeatletter \@addtoreset{equation}{section} \makeatother

\newtheorem{theorem}{Theorem}[section]
\newtheorem{lemma}{Lemma}[section]
\newtheorem{corollary}{Corollary}[section]

\newtheorem{proposition}{Proposition}[section]

\newcommand{\mdet}{\mathrm{det}}
\newcommand{\intd}{\displaystyle\int}
\newcommand{\Tr}{\mathrm{Tr}\,}
\newcommand{\Str }{\mathrm{Str}\,}
\newcommand{\Sdet}{\mathrm{Sdet}\,}
\newcommand{\E}{\mathbb{E}}

\newcommand{\Do}{D^{\mathrm{o}}}

\newcommand{\conjSupVec}[1]{#1^*}

\newcommand{\herm}{\mathcal{H}}

\newcommand{\diag}{\mathop{\mathrm{diag}}}

\begin{document}
\title{ Universality of the second correlation function of the deformed Ginibre ensemble}
\author{ Ievgenii Afanasiev\thanks{B. Verkin Institute for Low Temperature Physics and Engineering of the National Academy of Sciences of Ukraine, Kharkiv, Ukraine;
e-mail: afanasiev@ilt.kharkov.ua} \and
Mariya Shcherbina
\thanks{B. Verkin Institute for Low Temperature Physics and Engineering of the National Academy of Sciences of Ukraine, Kharkiv, Ukraine\,\&\,Institute for Advanced Study, Princeton, USA; e-mail: shcherbi@ilt.kharkov.ua} \and
 Tatyana Shcherbina
\thanks{ Department of Mathematics, University of Wisconsin - Madison, USA, e-mail: tshcherbyna@wisc.edu. This material is based upon work supported  in part by Alfred P. Sloan Foundation grant FG-2022-18916 and the National Science Foundation under grant DMS-2346379}}

\date{}
\maketitle

\begin{abstract}
We study the deformed complex Ginibre ensemble $H=A_0+H_0$, where $H_0$ is  the complex  matrix with iid Gaussian entries, and $A_0$ is some general $n\times n$ matrix 
(it can be random and in this case it is independent of $H_0$). Assuming rather general assumptions on $A_0$, we prove that the asymptotic local behavior of the second correlation function of the eigenvalues of such matrices in the bulk coincides with that for the pure complex Ginibre ensemble.

\end{abstract}

\section{Introduction}
Consider the complex deformed Ginibre ensemble, i.e. $n\times n$ matrices
 \begin{equation}\label{defG}
H=A_0+H_0,
\end{equation}
where $H_0$ is a complex Ginibre matrix with i.i.d. complex Gaussian entries $\{h_{ij}^{(0)}\}_{i,j=1}^n$ such that
\[
\E[h_{ij}^{(0)}]=0, \quad \E[|h_{ij}^{(0)}|^2]=1/n,\quad \E[(h_{ij}^{(0)})^2]=0.
\]
The deformation matrix $A_0$ is a general $n\times n$ matrix with complex entries (which can be random, and
in this case is independent of $H_0$). Such matrices has a lot of applications, in particular, in computational mathematics and communication theory.

The ensemble (\ref{defG}) and its more general case with non-Gaussian $H_0$ is extensively studied in mathematical literature. In particular, it is known from \cite{TVKr:10} that under the reasonable general assumption on the distribution of $A_0$  and for more general $H_0$ with iid but non-Gaussian entries, there exists a limiting spectral distribution $\mu$ of the eigenvalues of $H$.  If $A_0=0$, then
$\mu$ is the celebrated {\it circular law}, a uniform distribution on a unit disk in the complex plain.
But for $A_0\ne 0$, the limiting distribution is not necessary the circular law and  it hardly depends on $A_0$. The exact form of the support $D$ of the measure $\mu$ is not simple to describe, however, 
according to \cite{BoCap:14} (see also the review \cite{BoCh:12} for the history of the problem and  \cite{Z:22} and references therein for the more general case), under rather general conditions on $A_0$
it takes the nice form:
 \begin{align}\label{dD}
  D=\Big\{z: \int \lambda^{-1}d\nu_z(\lambda)\ge 1\Big\},
 \end{align}
 where $\nu_z$ is the limiting normalized eigenvalue counting measure of $ Y_0(z)=(A_0-z)(A_0-z)^*$.
We refer to the interior points inside $D$ as the {\it bulk} of the spectrum (and denote it as $\Do$), and to the points on the boundary $\partial D$ as the {\it edge} of the spectrum. 

The behaviour of individual eigenvalues, or the local eigenvalue statistics, is expected to be more universal, i.e. it does not depend on $A_0$. 
This is a part of a general Wigner-Dyson-Mehta universality conjecture 
stating that the behavior of local statistics of eigenvalues of random matrices is determined by the symmetry type of ensemble and largely does
not depend on the entries distribution or even their independence. The conjecture is going back to
the Wigner’s pioneering idea to model spectral statistics of complex quantum systems by those of simple random matrix ensembles
that respect the basic symmetries but otherwise may not be related at all to the initial quantum Hamiltonian.
The original conjecture concerned mostly the Hermitian and symmetric ensembles of random matrices and starting from the breakthrough
works of Erd$\ddot{\hbox{o}}$s, Yau, Schlein with co-authors (see \cite{ErY:book} and reference therein)  and Tao and Vu (see, e.g., \cite{TV:11})
is proved now for most classical random graphs and random matrix models. 

However, for the non-Hermitian random matrices with two-dimensional spectral distribution the local spectral statistics of eigenvalues is much less studied.
One of the main reasons why the non-Hermitian
spectral analysis is so difficult is because, unlike in the Hermitian case, the resolvent
$(H-z)^{-1}$ of a non-normal matrix is not effective to study eigenvalues near $z$ due to the instability of the non-Hermitian eigenvalues under perturbations.

The main useful tool to deal with  eigenvalues of  non-Hermitian matrices is Girko's logarithmic potential approach \cite{Gir:84}
based on the formula  expressing linear statistics of eigenvalues of $H$ in
terms of the log-determinant of the symmetrized matrix 
\begin{align}\label{Y}
Y(z)=(H-z)(H-z)^*.
\end{align}
For example, the linear eigenvalue statistics for any smooth, compactly supported test function $f$ can be expressed as
\begin{align}\label{Gir}
N_n[f]:=\sum\limits_{j=1}^n f(z_j)=
\frac{1}{4\pi }\int_{\mathbb{C}}\Delta_z f(z)\cdot \log|\mdet \,Y(z)| d^2z,
\end{align}
where $\Delta_{z}:=\frac{\partial^2}{\partial z\partial \bar z}$. Then, the $k$th  correlation function  $p_k(z_1,\dots z_k)$  can be recovered  as
\begin{align*}
p_k(z_1,\dots, z_k)=(4\pi n)^{-k}\Delta_{z_1}\dots\Delta_{z_k}\E\{\log\det Y(z_1)\dots\log\det Y(z_k)\}
\end{align*}

\textit{Bulk} universality conjecture of the local eigenvalue statistics states
that for any $z_0\in \Do$ uniformly in $\{\zeta_j\}_{j=1}^k$ varying in any compact set $K\subset \mathbb{C}$  we have
\begin{align}\label{Un_bulk}
\lim\limits_{n\to\infty}\rho^{-k}p_k(z_0+\zeta_1/(\rho n)^{-1/2},\dots,z_0+\zeta_k/(\rho n)^{-1/2})=\det \{K^{(b)}(\zeta_i,\zeta_j)\}_{i,j=1}^k,
\end{align}
where $\rho$ is a limiting density $p_1(z_0)$ and
\begin{equation*}
K^{(b)}(\zeta_1,\zeta_2)=\frac{1}{\pi}e^{-|\zeta_1|^2/2-|\zeta_2|^2/2+\zeta_1\bar \zeta_2}.
\end{equation*}
This means that the limit coincides with that obtained for the pure Ginibre ensemble. Similar statement (but with a different kernel $K^{(e)}(\zeta_1,\zeta_2)$) can be formulated 
for $z_0\in \partial D$ (so-called \textit{edge} universality).

The key point of Girko's logarithmic potential approach is that $Y(z)$ is now a Hermitian matrix and all tools and results developed for
Hermitian ensembles in the last years are available. The problem is that Girko's formula is much harder to analyse than the analogous expression
for the Hermitian case. In particular, it requires a good lower bound on the smallest singular value $\sigma_1(H-z)$
of $H-z$, a notorious difficulty which was already crucial in the proof of circular law and its more general analog \cite{TVKr:10}.
However, for the study of the asymptotic behavior of {\it individual } eigenvalues  much more precise control of $\sigma_1(H-z)$ is required.

Such control is hardly accessible for the standard random matrix techniques (see, e.g., the discussion in \cite{CiErS:ed} for details).
For the case of matrices with iid entries, after the long series of results,
bulk and edge universality for a general distribution  were obtained in the classical work of Tao and Vu \cite{TV:15} (see also reference therein for the history of the problem) under the assumption that  the common distribution 
of the entries is good enough (e.g., has all bounded moments), and its first four moments
match the first four moments of the standard Gaussian distribution.  The result holds also for the case of matrices with iid real entries.

The removing of this four moment matching condition happened very recently and required combination of the well-developed random
matrix techniques with the supersymmetry (SUSY) approach based on the representation of the determinant as an integral (formal) over the Grassmann variables. 
Combining this representation with the representation of an inverse determinant as  integral over the Gaussian complex field,
SUSY allows to obtain an integral representation for the main spectral characteristics such as averaged density of states, characteristic polynomials, 
correlation functions, as well as for elements of the resolvent moments, etc. 
Although the rigorous control of such SUSY integral representations can be difficult, 
SUSY approach is widely  used in the physics literature  (see e.g. reviews \cite{Ef}, \cite{M:00}) and was successfully applied rigorously to
the study of some Hermitian and non-Hermitian random matrices (see, e.g.,  \cite{FyoSom:96}, \cite{Fyo:18}, \cite{TSh:ChW}, \cite{Af:19}, \cite{SS:Un}, \cite{AfS:23}).

Using integral representation obtained by SUSY, the optimal control of $\sigma_1(H-z)$ for the Ginibre ensemble at the edge of the spectrum was achieved in \cite{CiErS:20} (the result for $A_0\ne 0$ was obtained in \cite{SS:Gin}). 
The  result of \cite{CiErS:20} was used in the subsequent work \cite{CiErS:ed}  to remove the four moment matching condition of \cite{TV:15} for the edge universality.
Very recently, in the work of Maltsev and Osman \cite{MalOs:23}, SUSY together with a partial Schur decomposition was applied to obtain
the bulk universality for the matrices (\ref{defG}) with a small Gaussian perturbation $\sqrt{t} H_0$ and matrices $A_0$  with a very good control of
the resolvent of $(A_0-z)(A_0-z)^*+\varepsilon^2$ (basically, they need $A_0$ to satisfy the two-resolvent local law,  see \cite{AEK:21}, \cite{CiErS:CLT_mes}).
Similar result  was obtained in \cite{Zhang:24}  for diagonal $A_0$  with only a finite number of different parameters $\{a_j\}$ on its diagonal. The result of \cite{MalOs:23} combined with the local law allowed to
remove the four moment matching condition of \cite{TV:15} for the bulk universality. The matrices with iid real entries were similarly treated recently in \cite{Os:24}, \cite{DubY:24}.

The main aim of the current paper is to prove bulk universality (\ref{Un_bulk})  of the second correlation function $p_2(z_1,z_2)$ for the ensemble (\ref{defG}) with a rather general deformation $A_0$  which is not restricted to be normal, to satisfy the local law, or to have a limiting spectrum in the unit disc. 
We would like to mention that we expect that our proof can be transposed directly to the case of $k$th correlation function if we replace $2\times 2$ matrices by
$k\times k$ in the argument.
The only exception is technical Lemma \ref{l:1} (see Section \ref{s:3}), where the matrix dimension is used essentially.

To get the integral representation of $p_2(z_1,z_2)$, we apply SUSY. Following \cite{FyoSom:96},  instead of the $\E\{\log\det Y(z_1)\log\det Y(z_2)\}$ with $Y(z)$ of  (\ref{Y}) we consider the following generating functional
\begin{align}\label{Z}
&\mathcal{Z}(\zeta,\zeta',\hat\varepsilon, \hat\varepsilon') = \E\biggl\lbrace \prod\limits_{j = 1}^{2} \frac{\det(Y(z_j) + \varepsilon_j^2)}{\det(Y(z'_j) + \varepsilon'^2)} \biggr\rbrace,\quad z_j=z_0+\zeta_j/\sqrt n,\quad z_j'=z_0+\zeta_j'/\sqrt n,\\
&\zeta=\mathrm{diag}\,\{\zeta_1,\zeta_2\},\quad \zeta'=\mathrm{diag}\,\{\zeta_1',\zeta_2'\}\quad \hat\varepsilon=\mathrm{diag}\,\{\varepsilon_1,\varepsilon_2\},\quad \hat\varepsilon'=\varepsilon' I_2.\notag
\end{align}
Here we choose $z_1,z_2$ in the $n^{-1/2}$-neighbourhood of some point $z_0\in D$.
It is easy to check that
\begin{align}\label{der}
&p_2(z_0+\zeta_j/\sqrt n,z_0+\zeta_j'/\sqrt n)\\
&\hskip2cm =\lim_{\epsilon\to 0} (4\pi )^{-2}\partial_1  \partial_2\frac{\partial}{\partial\bar\zeta_1}\frac{\partial}{\partial\bar\zeta_2}
\mathcal{Z}(\zeta,\zeta',n^{-1}\epsilon I_2, n^{-1}\epsilon' I_2)
 \Big|_{\zeta_1=\zeta_1',\bar\zeta_1=\bar\zeta_1',\hat\epsilon=\hat\epsilon'=
 \epsilon I_2}, \notag
 \end{align}
where in order to simplify formulas we use notations:
\begin{align}\label{pa_1,2}
& \partial_1 =\frac{\partial}{\partial\zeta_1}+\frac{\partial}{\partial\zeta_1'},\quad 
 \partial_2 =\frac{\partial}{\partial\zeta_2}+\frac{\partial}{\partial\zeta_2'}.
\end{align}

Now suppose $A_0$ satisfies the following conditions

\textit{Assumptions (A1)-(A3)}:
\begin{enumerate}
\item[(A1)] There are some $M,d>0$ such that
\begin{align*}
\mathrm{Prob}\Big\{n^{-1}\sum_{i,j=1}^n|A_{0,ij}|^2<M\Big\}\ge 1- n^{-1-d}.
\end{align*}
\item[(A2)] 
For almost all $z$ normalized counting measure $\nu_{z,n}$ of eigenvalues  of the matrix\\ $Y_0(z):=(A_0-z)(A_0-z)^*$ converges, as $n\to\infty$, to some limiting
measure $\nu_z$;

\item[(A3)]  For any $z\in \Do$ (see (\ref{dD})) there are $d_1(z)>0,\varepsilon(z)>0$ such that

\begin{align*}  
\mathrm{Prob}\Big\{n^{-1}\Tr \big(Y_0(z)+\varepsilon^2(z)\big)^{-1}>1+d_1(z)\Big\}>1-C'_\epsilon n^{-1-d}.
\end{align*}
\end{enumerate}
Notice that in the case when $A_0$ is non-random,  assumptions (A1)-(A3) mean that starting from some $n$ the inequalities
in (A1)-(A3) (which in the random case we want to have with  probability higher than $1-n^{-1-d}$) are valid.

According to \cite{TVKr:10}, the assumptions (A1)-(A2) guarantee that there exists a limiting distribution $\mu$ of the eigenvalues of (\ref{defG}).   In addition, for almost all $z\in\mathbb{C}$ there exists
a non-random probability measure $\eta_z$ on $\mathbb{R}$ which is a limit of the normalized counting measure of eigenvalues of $Y(z)$
defined in (\ref{Y}) (see \cite{DSil:07}). 

To formulate the main result of the paper we need a few notations.
Denote by $u_*$ a positive solution of the equation
\begin{align}\label{eq_u}
 & \int(u_*^2+\lambda^2 )^{-1}d\nu_{n,z}(\lambda^2)=1\,\Longleftrightarrow  n^{-1} \Tr G=1,\\
&G:=(A_zA_z^*+u_*^2)^{-1},\quad A_z=A_0-z_0.
\label{G}\end{align}
Notice that this solution is unique and the condition (A3) guarantees that $u_*>\varepsilon(z)>0$ with high probability.

We introduce also notations for scalar characteristics of the matrix $A_z$ which appear in our consideration below:
\begin{align}\notag
g_2=&n^{-1}\Tr G^2,\quad k_A=n^{-1}\Tr A_zG ,\quad h_A=n^{-1}\Tr A_zG^2,\quad f_A=n^{-1}\Tr (A_zG)^2,\\
G_*=&(A_z^*A_z+u_*^2)^{-1}, \quad \rho=n^{-1}u_*^2\Tr GG_*+g_2^{-1}|h_A|^2.
\label{rho}\end{align}

\begin{theorem}\label{t:1} Let $A_0$ satisfies assumptions  (A1)-(A3) and $z_0\in \Do$ (see (\ref{dD})). Then  for any $\tilde M>0$ uniformly in 
$|\zeta_{1}|<\tilde M$, $|\zeta_{2}|<\tilde M$ we have 
\begin{align}\label{t1.1}
&p_2(z_0+\zeta_1/n^{1/2},z_0+\zeta_2/n^{1/2})\stackrel{w}{\longrightarrow}
 \rho^2(1-e^{-\rho|\zeta_1-\zeta_2|^2}),
\end{align}
where $\mathcal{Z}$ and $\rho$ are defined by (\ref{Z})  and (\ref{rho}) respectively.
\end{theorem}

An important ingredient of our proof  is a following proposition, which allows us to express a weak limit of $p_2(z_0+\zeta_1/n^{1/2},z_0+\zeta_2/n^{1/2})$
in terms of the generating functional $\mathcal{Z}(\zeta,\zeta',\frac{\epsilon}{n} I_2, \frac{\epsilon'}{n} I_2)$.

\begin{proposition}\label{p:logZ} For any $z_{l}=z_0+n^{-1/2}\zeta_{l}$, $l=1,2$ with $\{\zeta_{l}\}_{l=1}^2$ varying on the compact set and any $\epsilon_1>0,\epsilon_2>0$
\begin{align}\label{p.logZ}
&\Big|\mathbb{E}\{\log \det Y(z_1)\log \det Y(z_2)\}\\&\hskip2cm-\mathbb{E}\Big\{\log\det\Big( Y(z_1)+\Big(\frac{\epsilon_1}{n}\Big)^2\Big)
\log\det\Big ( Y(z_2)+\Big(\frac{\epsilon_2}{n}\Big)^2\Big)\Big\}\Big|\le C\epsilon_1\epsilon_2.
\notag\end{align}
\end{proposition}
The proof of the proposition is given in  Section \ref{s:5}.

We prove Theorem \ref{t:1} in  a few steps. First, in Section \ref{s:2} we derive the integral representation for 
$\mathcal{Z}(\zeta,\zeta',\hat\varepsilon, \hat\varepsilon')$ using SUSY approach. 
For the reader convenience, the brief outline of  the basic formulas of SUSY techniques is given in Section 2.1.  The obtained integral representation
(see (\ref{repr_fin})) contains a large parameter $n$  at the exponent, hence we analyse it using a saddle-point method. Section \ref{s:3} is devoted
to the proof of the fact that only a small neighbourhood of a saddle-point contributes to the integral. Unfortunately, the formula  obtained in Section \ref{s:3}
after the saddle-point analysis (see (\ref{repr_Z})) still contains an additional multiplier $n^4$ in front of the integral. This suggests that non-zero contribution to the integral is given by the $8$-th order of the Taylor expansions at the neighbourhood of the saddle-point of all functions under the integral.
Since our functions depend on 16 variables, the corresponding expansion looks  too complicated, and we prefer to use a different method.
In Section \ref{s:4} we do a number of changes  in Grassmann and scalar variables which allow us to transform (\ref{repr_Z}) into an integral
which does not have additional factor $n^\alpha$ in front of it. In addition, the obtained integral  (see (\ref{final})) depends
on $A_0$ only through 2 scalar parameters $\rho$, $u_*^2/g_2$ (see (\ref{rho})). Since the same representation is valid for $A_0=0$ with a
different $|z_0'|<1$ and $\zeta$ (see the argument at the end of Section \ref{s:4.1}), we conclude that the result of the limiting and differentiation
procedure (\ref{der}) is universal (modulo two above scalar parameters) for any $A_0$ satisfying assumption (A1)-(A3).
 Combined with Proposition \ref{p:logZ}, this argument completes the proof of Theorem \ref{t:1}.

%

\section{Integral representation}\label{s:2}
In this section we derive an integral representation for $\mathcal{Z}(\zeta,\zeta',\hat\varepsilon, \hat\varepsilon')$ using supersymmetry techniques. 
For the reader convenience, the brief outline of the basic formulas of SUSY techniques is given in Section 2.1. 
More detailed information of the techniques and its applications  can be found, e.g.,
in \cite{M:00},  \cite{Ef}.
\subsection{SUSY techniques: basic formulas}
Consider two sets of formal variables
$\{\psi_j\}_{j=1}^n,\{\overline{\psi}_j\}_{j=1}^n$ satisfying the anticommutation relations:
\begin{equation}\label{anticom}
\psi_j\psi_k+\psi_k\psi_j=\overline{\psi}_j\psi_k+\psi_k\overline{\psi}_j=\overline{\psi}_j\overline{\psi}_k+
\overline{\psi}_k\overline{\psi}_j=0,\quad j,k=1,\ldots,n.
\end{equation}
These two sets of variables $\{\psi_j\}_{j=1}^n$ and $\{\overline{\psi}_j\}_{j=1}^n$ generate the Grassmann
algebra $\mathfrak{A}$. Taking into account that anticommutation relations imply $\psi_j^2=\bar\psi_j^2=0$,  all elements of $\mathfrak{A}$
are polynomials of $\{\psi_j\}_{j=1}^n$ and $\{\overline{\psi}_j\}_{j=1}^n$ of degree at most one
in each variable. If such polynomial has only monomials of even power, then it is clearly a commuting element of $\mathfrak{A}$, and we 
call such elements \textit{even}. If the polynomial has only monomials of odd power, we call it an \textit{odd} element (such elements anticommutes with each other).

Define also functions of the Grassmann variables. Let $\chi$ be an element of $\mathfrak{A}$, i.e.
\begin{equation}\label{chi}
\chi=a+\sum\limits_{j=1}^n (a_j\psi_j+ b_j\overline{\psi}_j)+\sum\limits_{j\ne k}
(a_{j,k}\psi_j\psi_k+
b_{j,k}\psi_j\overline{\psi}_k+
c_{j,k}\overline{\psi}_j\overline{\psi}_k)+\ldots.
\end{equation}
Given the sufficiently smooth function $f$, we define by $f(\chi)$ the element of $\mathfrak{A}$ obtained by substituting $\chi-a$
in the Taylor series of $f$ at the point $a$:
\[
f(\chi)=a+f'(a)(\chi-a)+\dfrac{f''(a)}{2!}(\chi-a)^2+\ldots
\]
Since $\chi$ is a polynomial of $\{\psi_j\}_{j=1}^n$,
$\{\overline{\psi}_j\}_{j=1}^n$ of the form (\ref{chi}), according to (\ref{anticom}) there exists such
$l$ that $(\chi-a)^l=0$, and hence the series terminates after a finite number of terms.

Following Berezin \cite{Ber}, we define the operation of
integration with respect to the anticommuting variables in a formal
way:
\begin{equation*}
\intd d\,\psi_j=\intd d\,\overline{\psi}_j=0,\quad \intd
\psi_jd\,\psi_j=\intd \overline{\psi}_jd\,\overline{\psi}_j=1,
\end{equation*}
and then extend the definition to the general element of $\mathfrak{A}$ by
the linearity. A multiple integral is defined to be a repeated
integral. Assume also that the ``differentials'' $d\,\psi_j$ and
$d\,\overline{\psi}_k$ anticommute with each other and with the
variables $\psi_j$ and $\overline{\psi}_k$. More precisely, if
$$
f(\psi_1,\ldots,\psi_k)=p_0+\sum\limits_{j_1=1}^k
p_{j_1}\psi_{j_1}+\sum\limits_{j_1<j_2}p_{j_1,j_2}\psi_{j_1}\psi_{j_2}+
\ldots+p_{1,2,\ldots,k}\psi_1\ldots\psi_k,
$$
then
\begin{equation*}
\intd f(\psi_1,\ldots,\psi_k)d\,\psi_k\ldots d\,\psi_1=p_{1,2,\ldots,k}.
\end{equation*}
It is easy to see that one can perform a linear change of variables in the Grassmann integrals. Namely,  
suppose $$\chi=\begin{pmatrix}
\chi_1\\ \vdots \\ \chi_k
\end{pmatrix},\quad d\chi=d\chi_k\ldots d\chi_1,$$ and
$A$ is invertible complex matrix. Then if $\chi=A\zeta$,  it is easy to check that
\begin{equation}\label{Gr_change}
\int f(\chi)d\chi=\frac{1}{\det A}\int f(A\zeta)d\zeta.
\end{equation}
Note that the ``Jacobian" (which is called a ``Berezinian" for the change of Grassmann variables) is an {\it inverse} determinant in contrast to the
the determinant for the usual integral.

One can also perform a linear shift of variables $\chi+\psi\to \chi$ by writing 
\begin{equation}\label{Gr_shift1}
\int f(\chi+\psi)d\psi=\int f(\chi)d\chi
\end{equation}
for any $k$-dimension Grassmann vector $\psi$  with odd (anticommuting) elements (which does non depend on $\chi$). Moreover, one can easily show that if 
$f:\mathbb{R}^k\to \mathbb{R}$ as a sufficiently smooth function with finite support (or exponential decay at $\infty$), then for the usual real integral
\begin{equation}\label{Gr_shift2}
\displaystyle\int_{\mathbb{R}^k} f(x+a)dx=\int_{\mathbb{R}^k} f(x)dx
\end{equation}
for any $k$-dimension Grassmann vector $a$  with even (commuting) elements  independent of $x$.

   Let $A$ be an ordinary Hermitian matrix with a positive real part. The following Gaussian
integral is well-known
\begin{equation}\label{G_C}
\intd \exp\Big\{-\sum\limits_{j,k=1}^nA_{jk}z_j\overline{z}_k\Big\} \prod\limits_{j=1}^n\dfrac{d\,\Re
z_jd\,\Im z_j}{\pi}=\dfrac{1}{\mdet A}.
\end{equation}
One of the important formulas of the Grassmann variables theory is the analog of this formula for the
Grassmann algebra (see \cite{Ber}):
\begin{equation}\label{G_Gr}
\int \exp\Big\{-\sum\limits_{j,k=1}^nA_{jk}\overline{\psi}_j\psi_k\Big\}
\prod\limits_{j=1}^nd\,\overline{\psi}_j d\,\psi_j=\mdet A,
\end{equation}
where $A$ now is any $n\times n$ matrix.

Let
\[
F=\left(\begin{array}{cc}
A&\chi\\
\eta&B\\
\end{array}\right),\quad
\]
where $A$ and $B$ are Hermitian complex $k\times k$ matrices such that $\Re B>0$ and $\chi$, $\eta$ are $k\times k$ matrices
of independent
anticommuting Grassmann variables, and let
\[
\theta=(\psi_{1},\ldots,\psi_{k},x_{1},\ldots,x_{k})^t,
\]
where $\{\psi_j\}_{j=1}^k$ are independent Grassmann variables and $\{x_j\}_{j=1}^k$ are complex variables.
Combining (\ref{G_C}) -- (\ref{G_Gr}) we obtain (see \cite{Ber})
\begin{equation}\label{G_comb}
\intd \exp\{-\theta^*F\theta\}\prod\limits_{j=1}^kd\overline{\psi}_j\, d\psi_j \prod\limits_{j=1}^k\dfrac{\Re x_j\Im
x_j}{\pi}=\Sdet\, F,
\end{equation}
where
\begin{equation}\label{sdet}
 \Sdet\, F=\dfrac{\det\,(A-\chi\, B^{-1}\,\eta)}{\det\, B}.
\end{equation}
Note that this definition allows to maintain the usual properties of determinants such as $\Sdet\, (F_1F_2)=\Sdet\, F_1\cdot \Sdet\, F_2$, etc.
In addition, if we define
\begin{equation}\label{Str}
\Str\,F=\Tr A -\Tr B, 
\end{equation}
then we have a usual properties of trace such as 
\begin{align*}
&\Str F_1F_2=\Str F_2F_1,\\
&\log \Big(\Sdet F\Big)=\Str\Big(\log F\Big).
\end{align*}
We will need also the following Hubbard-Stratonovich transform formulas based on Gaussian integration:
\begin{align*}
&e^{ab}=\pi^{-1} \int e^{a \bar u+b u-\bar u u} d\bar u\, du,\\ 
&e^{-\rho\tau}=\int e^{\rho \chi+\tau\eta+\chi\eta} d\eta d\chi.
\end{align*}
Here $a, b$ can be complex numbers or any even (commuting) elements of $\mathfrak{A}$,
and $\rho,\tau$ are the odd (anticommuting) elements of $\mathfrak{A}$.
Applying these formulas multiple times one can get the Hubbard-Stratonovich transform formulas in the matrix form. Let $A, B$ are $p\times p$ matrices of complex numbers or even elements of $\mathfrak{A}$,
and $R,T$ are $p\times p$ matrices of odd elements of $\mathfrak{A}$. Then
\begin{align}\label{Hub_C}
&e^{\Tr AB}=\int e^{\Tr AW^*+\Tr BW-\Tr WW^*} dW,\\ \label{Hub_Gr}
&e^{\Tr  RT}=\int e^{-\Tr \nu\nu^*+\Tr\nu R+\Tr \nu^* T} d\nu.
\end{align}
Here $W$ is $p\times p$ complex matrix, $\nu=\{\nu_{j,k}\}_{j,k=1}^p$, $\nu^*=\{\bar\nu_{k,j}\}_{j,k=1}^p$ are $p\times p$ matrices of independent Grassmann variables and
\[
d\nu=\prod\limits_{j,k=1}^pd\nu_{kj}d\bar{\nu}_{jk} ,\quad dW=\prod\limits_{j,k=1}^p\frac{d\Re w_{jk} d\Im w_{jk}}{\pi}.
\]

\subsection{Derivation of the integral representation for $\mathcal{Z}(\zeta,\zeta',\hat\varepsilon, \hat\varepsilon')$}
The main aim of the section is to prove the following proposition:
\begin{proposition}\label{p:repr} Given $\mathcal{Z}(\zeta,\zeta',\hat\varepsilon, \hat\varepsilon')$ of (\ref{Z}) with $\hat\varepsilon=\hat \epsilon/n$, $\hat\varepsilon'=\hat\epsilon'/n$ we have
\begin{align}\label{repr_fin}
\mathcal{Z}(\zeta,\zeta',\hat\epsilon/n, \hat\epsilon'/n) = &\frac{ n^{12}}{32 \pi^{5}}\int\limits_{\herm_2^+} \,dR\,dR_2
 \int\limits_{\Im T= u_*} dT \int\limits_{\herm_2}dS
  \int\limits_{\herm_2^+}  d\Lambda\int\limits_{U(2)} dU  \int d\nu \,\,\\ 
  & \times e^{\mathcal{L}(\tilde Q,T,S,R)} 
  E_{*1}(\hat\epsilon,\Lambda,U)E_{*2}(\hat\epsilon',R,R_2), 
\notag\end{align}
where $dT$ above  means the integration  over  all independent entries of $T=iu_*I_2+T'$ ($T'=T'^*$), 
with $u_*$ defined in (\ref{eq_u}), $dU$ is a Haar measure over unitary group $U(2)$, and
\begin{equation}\label{dnu}
d\nu=\prod\limits_{l=1}^2\prod\limits_{j,k=1}^2d\nu^{(l)}_{kj}d\bar{\nu}^{(l)}_{jk}
\end{equation}
is an integration with respect  to $2\times 2$ matrices $\nu_l$, $\nu_l^*$, $l=1,2$ of independent Grassmann variables.
Here
\begin{equation}\label{L}
\mathcal{L}(\hat Q,T,S,R):= \log\,\Sdet \tilde Q-n \Tr \lbrack R^2+ 2iR T+ \Lambda^2 \rbrack  -\Tr (\nu_1\nu_1^* + \nu_2\nu_2^*), 
\end{equation}
\begin{equation}\label{hat_Q}
\hat Q = \begin{pmatrix}
\Lambda\otimes I_n & i\mathcal A(\zeta) & n^{-1/2}\nu_1 \otimes I_n & 0 \\
i\mathcal A(\zeta_U)^* &\Lambda \otimes I_n & 0 & n^{-1/2}\nu_2 \otimes I_n \\
n^{-1/2}\nu_2^* \otimes I_n & 0 & i(T+S) \otimes I_n & i\mathcal A(\zeta'_J)\\
0 & n^{-1/2}\nu_1^* \otimes I_n &i\mathcal A(\zeta'_J)^* & i(T-S)  \otimes I_n
\end{pmatrix}
\end{equation} 
with
\begin{align}
\label{cal_A}
&\mathcal A(\zeta) =  I_2 \otimes A_z + \frac{1}{\sqrt{n}}\zeta \otimes I_n, \quad A_z = A_0 - z_0,\\
& \zeta_{U}=U\zeta U^*,\quad \zeta'_J=J^{-1}\zeta' J,\quad J=R_2^{1/2}R^{1/2},
\label{zeta_U}\end{align}
\begin{align}\notag
E_{*1}(\hat\epsilon,\Lambda,U)=& (\Tr \Lambda)^2 \det \Lambda \exp\Big\{-\Tr (\Lambda U\hat\epsilon+\hat\epsilon U^*\Lambda )\Big\},\notag\\
E_{*2}(\hat\epsilon',R,R_2)=& \dfrac{(\Tr R)^2 \det R}{ \det\nolimits^{2} R_2} \, \exp\Big\{- \Tr\big(\hat\epsilon'R_2+\hat\epsilon' R_2^{-1}R^2\big)\Big\}.
\label{E_*1,2}\end{align}
\end{proposition}
\textit{Proof.} 
Introduce the sets of Grassmann and complex variables:
\begin{align*}
&\Psi_j^{(l)}=(\psi_{j1}^{(l)},\dots,\psi_{jn}^{(l)})^t,\quad \bar\Psi_j^{(l)}=(\bar\psi_{j1}^{(l)},\dots,\bar\psi_{jn}^{(l)})^t,\quad l=1,2, \quad j=1,2 \quad -\quad \hbox{Grassmann};\\
&X_j^{(l)}=(x_{j1}^{(l)},\dots,x_{jn}^{(l)})^t,\quad \bar X_j^{(l)}=(\bar x_{j1}^{(l)},\dots,\bar x_{jn}^{(l)})^t,\quad l=1,2,\quad j=1,2 \quad -\quad \hbox{complex}.\notag
\end{align*}
and use the standard  linearisation  formula
\begin{align*}
\det(Y(z)+\varepsilon^2)=\det \tilde Y(z,\varepsilon), \quad \tilde Y(z,\varepsilon)
=\left(\begin{array}{cc}-\varepsilon&i(H-z)\\i(H-z)^*&-\varepsilon\end{array}\right).
\end{align*}
Since all eigenvalues of $\tilde Y(z,\varepsilon)$ have the real part $-\varepsilon$, we can apply (\ref{G_C}), (\ref{G_Gr}) to write for $j=1,2$
\begin{align*}
&(\det(Y(z_j')+(\varepsilon_j')^2))^{-1}=\int d\bar X_j dX_j\exp\{ X_j^*\tilde Y(z_j',\varepsilon_j')X_j\},\\
&\det(Y(z_j)+\varepsilon_j^2)=\int d\bar \Psi_j d\Psi_j \exp\{- \Psi_j^*\tilde Y(z_j,\varepsilon_j)\Psi_j\},
\end{align*}
with  $2n$-dimensional column-vectors $$X_j=\left(\begin{matrix}X_j^{(1)}\\ X_j^{(2)}\end{matrix}\right),\quad \Psi_j=\left(\begin{matrix}\Psi_j^{(1)}\\\Psi_j^{(2)}\end{matrix}\right),$$ 
\begin{align*}
d\bar X_j dX_j=\prod_{\alpha=1}^n \prod_{l=1}^2\dfrac{d\bar x_{j\alpha}^{(l)} dx_{j\alpha}^{(l)}}{\pi},\quad d\bar \Psi_j d\Psi_j =\prod_{\alpha=1}^n \prod_{l=1}^2 
d\bar \psi_{j\alpha}^{(l)} d\psi_{j\alpha}^{(l)}.
\end{align*}
Thus we get
\begin{equation*}
\mathcal{Z}(\zeta,\zeta',\hat\varepsilon, \hat\varepsilon') = \int Z(X) \,\tilde Z(\Psi) \,E(\Psi,X) \,dX d\Psi ,
\end{equation*}
where
\begin{align}\label{ZZE}
&Z(X)=\prod_{j=1}^2\exp\Big\{i\Tr(A_0-z_j')^*X_j^{(1)} X_j^{(2)*}+i\Tr (A_0-z_j')X_j^{(2)}X_j^{(1)*}-\varepsilon'_j \Tr X_j^*X_j\Big\},\\ \notag
&\tilde Z(\Psi)=\prod_{j=1}^2\exp\Big\{-i\Tr (A_0-z_j)^*\Psi_j^{(1)}\Psi_j^{(2)*}-i\Tr (A_0-z_j)\Psi_j^{(2)}\Psi_j^{(1)*}+\varepsilon_j \Tr \Psi_j^* \Psi_j\Big\},
\end{align}
\begin{align*}
&E(\Psi, X) = \E \Bigl\lbrace \exp \Big\{ i\Tr H \big(X^{(2)}X^{(1)*}+\Psi^{(2)}\Psi^{(1)*}\big)+i\Tr H^* \big(X^{(1)}X^{(2)*}+\Psi^{(1)}\Psi^{(2)*}\big)\Big\}\Bigr\rbrace,
\end{align*}
with $n\times 2$ matrices
\[
X^{(l)}=(X^{(l)}_1,X^{(l)}_2),\quad \Psi^{(l)}=(\Psi^{(l)}_1,\Psi^{(l)}_2),\quad l=1,2
\]
and
\[
 dX=\prod_{j=1}^2 d\bar X_j dX_j,\quad d\Psi =\prod_{j=1}^2 d\bar \Psi_j d\Psi_j.
\]
Taking the expectation with respect to $H$ we get
\begin{align} \label{Z_av}
E(\Psi, X) &= \exp \Bigl\lbrace -\frac{1}{n} \Tr \Big(X^{(2)}X^{(1)*}+\Psi^{(2)}\Psi^{(1)*}\Big) \Big(X^{(1)}X^{(2)*}+\Psi^{(1)}\Psi^{(2)*}\Big)  \Bigr\rbrace\\ \notag
&=\exp \Bigl\lbrace -\frac{1}{n} \Bigl\lbrack \Tr X^{(1)*}X^{(1)}X^{(2)*} X^{(2)}-\Tr \Psi^{(1)*}\Psi^{(1)}\Psi^{(2)*} \Psi^{(2)} \\ \notag
&-\Tr \Psi^{(1)*}X^{(1)}X^{(2)*}\Psi^{(2)}+\Tr X^{(1)*}\Psi^{(1)}\Psi^{(2)*} X^{(2)}\Bigr\rbrack \Bigr\rbrace.
\end{align}
To explain an application of the Fourier transform,  recall that,
if we have functions $F,\sigma:\mathbb{R}^n\to \mathbb{R}$ and $f:\mathbb{R}\to \mathbb{R}$,
such that $F\in L_1(\mathbb{R}^n)\cap L_2(\mathbb{R}^n)$ and $f\in L_1(\mathbb{R})\cap L_2(\mathbb{R})$
and  $f(\sigma(X))\in L_1(\mathbb{R}^n)\cap L_2(\mathbb{R}^n)$,    then
\begin{equation*}
\int_{\mathbb{R}^n}F(X)f(\sigma(X))dX=(2\pi)^{-1}\int_{\mathbb{R}^n}F(X)\int \int e^{it(\sigma(X)-r)}f(r)dt drdX.
\end{equation*}
In our case we consider an analogue of  "condition function"  $\sigma$ of the matrix form
\begin{equation}\label{cond_X}
R_l=n^{-1} X^{(l)*}X^{(l)},\quad l=1,2,
\end{equation}
and then use corresponding matrix Fourier transform:
\begin{align}\label{Four}
& \int Z(X)\,\tilde Z(\Psi)E(\Psi,X) dX=
 \frac{n^8}{(2\pi)^8}\int dX\,Z(X)\tilde Z(\Psi)\int dT_1dT_2\int dR_1 dR_2\\
 &\exp\Big\{ -i\Tr (nR_1-  X^{(1)*}X^{(1)})T_1 -i\Tr(nR_2-X^{(2)*}X^{(2)})T_2  -n\Tr R_1R_2\Big\}\notag\\
 &\exp\Big\{n^{-1}\Big(\Tr \Psi^{(1)*}\Psi^{(1)}\Psi^{(2)*}\Psi^{(2)} +\Tr \Psi^{(1)*}X^{(1)}X^{(2)*}\Psi^{(2)}+\Tr \Psi^{(2)*} X^{(2)}X^{(1)*}\Psi^{(1)}\Big) \Big\}.
\notag\end{align}
To deal with $\Psi$-integration, we now apply the Hubbard--Stratonovich transformation (\ref{Hubb-Str}). Namely, we write
\begin{align}\label{Hubb-Str}
e^{n^{-1}\Tr \Psi^{(1)*}X^{(1)}X^{(2)*}\Psi^{(2)}} = \int 
e^{ -\Tr \nu_1\nu_1^* 
+ n^{-1/2} \Tr \nu_1 \Psi^{(1)*}X^{(1)} + n^{-1/2}\Tr \nu_1^*X^{(2)*}\Psi^{(2)}} d\nu_1, \\ \notag
e^{n^{-1}\Tr \Psi^{(2)*} X^{(2)}X^{(1)*}\Psi^{(1)}} = \int 
e^{ -\Tr \nu_2\nu_2^* 
+ n^{-1/2} \Tr \nu_2 \Psi^{(2)*}X^{(2)} + n^{-1/2}\Tr \nu_2^*X^{(1)*} \Psi^{(1)}} d\nu_2, \\ \notag
e^{n^{-1}\Tr \Psi^{(1)*}\Psi^{(1)}\Psi^{(2)*} \Psi^{(2)}} = n \int 
 e^{ -n \Tr W^*W 
  + \Tr W \Psi^{(1)*}\Psi^{(1)}+\Tr W^*\Psi^{(2)*} \Psi^{(2)}} dW,
\end{align}
where for $l=1,2$
\begin{equation*}
\nu_l=\left(\begin{matrix}
\nu_{11}^{(l)}&\nu_{12}^{(l)}\\
\nu_{21}^{(l)}&\nu_{22}^{(l)}
\end{matrix}\right), \quad  \nu_l^*=\left(\begin{matrix}
\bar \nu_{11}^{(l)}&\bar \nu_{21}^{(l)}\\
\bar \nu_{12}^{(l)}&\bar \nu_{22}^{(l)}
\end{matrix}\right)
\end{equation*}
are $2\times 2$ matrices of independent Grassmann variables, 
\begin{equation*}
W=\left(\begin{matrix}
w_{11}&w_{12}\\
w_{21}&w_{22}
\end{matrix}\right)
\end{equation*}
is a $2\times 2$ matrix of independent complex variables, and
\[
d\nu_l=\prod\limits_{j,k=1}^2d\nu^{(l)}_{kj}d\bar{\nu}^{(l)}_{jk} ,\quad dW=\prod\limits_{j,k=1}^2\frac{d\bar{w}_{jk} dw_{jk}}{\pi}.
\]
Combining  \eqref{Hubb-Str} and (\ref{Four}), we obtain
\begin{align}\label{Z after H-S}
&\mathcal{Z}(\zeta,\zeta',\hat\varepsilon, \hat\varepsilon') = \frac{n^{12}}{(2\pi)^8} \int\limits_{\herm_2^+} d R_1 dR_2 \int\limits_{\herm_2} dT_1 dT_2 
\int dW \int d\nu\\ \notag
&\times\exp\{-\Tr (\nu_1\nu_1^* + \nu_2 \nu_2^*) - n\Tr WW^* -in\Tr(R_1 T_1 + R_2 T_2) -n\Tr R_1R_2\rbrace \\  \notag
&\times\int d\Psi dX \exp \lbrace {\tilde{\Psi}}^* Q\tilde{\Psi} \rbrace,
\end{align}
where
 $\tilde{\Psi}$, $\conjSupVec{\tilde{\Psi}}$ are  super-vectors of the form
\begin{align*}
&\tilde{\Psi} =(\Psi^{(1)}_1, \Psi^{(1)}_2,\Psi^{(2)}_1, \Psi^{(2)}_2, X^{(1)}_1, X^{(1)}_2, X^{(2)}_1, X^{(2)}_2)^t,\\
&\tilde{\Psi}^* = (\bar\Psi^{(1)}_1, \bar\Psi^{(1)}_2,\bar\Psi^{(2)}_1, \bar\Psi^{(2)}_2, \bar X^{(1)}_1, \bar X^{(1)}_2, \bar X^{(2)}_1, \bar X^{(2)}_2),\\
&d\nu=d\nu_1\,d\nu_2,
\end{align*}
and  $Q$ is a $8n \times 8n$ super-matrix of the form
\begin{align}\label{Q_start}
Q = \begin{pmatrix}
(W - \hat\varepsilon) \otimes I_n &i\mathcal A(\zeta) & n^{-1/2} \nu_1 \otimes I_n & 0 \\
i\mathcal A(\zeta)^* & (W^* - \hat\varepsilon) \otimes I_n & 0 & n^{-1/2} \nu_2 \otimes I_n \\
n^{-1/2} \nu_2^* \otimes I_n & 0 & (iT_1 - \hat\varepsilon') \otimes I_n & i\mathcal A(\zeta') \\
0 & n^{-1/2} \nu_1^* \otimes I_n & i\mathcal A(\zeta')^* & (iT_2 - \hat\varepsilon') \otimes I_n
\end{pmatrix},
\end{align}
where  $\zeta = \diag \lbrace \zeta_1, \zeta_2 \rbrace$, $\zeta' = \diag \lbrace \zeta'_1, \zeta'_2 \rbrace$, and $\mathcal A(\zeta')$
is defined in (\ref{cal_A}).

Integrating with respect to $d\Psi dX$ according to \eqref{G_comb}, we get
\begin{equation}\label{sdetQ}
\begin{split}
\int d\Psi dX \exp \lbrace\tilde{\Psi}^* Q\tilde{\Psi}\rbrace  = \Sdet Q. 
\end{split}
\end{equation}
This implies
\begin{align*}
&\mathcal{Z}(\zeta,\zeta',\hat\varepsilon, \hat\varepsilon')= \frac{n^{12}}{(2\pi)^8} \int\limits_{\herm_2^+} d R_1 dR_2 \int\limits_{\herm_2} dT_1 dT_2 \int dW \int d\nu
\\ \notag
&\times \exp \lbrace\log\, \Sdet Q -in\Tr(R_1 T_1 + R_2 T_2) -n \Tr R_1R_2-\Tr (\nu_1 \nu_1^* + \nu_2 \nu_2^*) - n\Tr WW^*\rbrace 
\end{align*}
Further, change the variables 
\begin{align*}
T_1 &\to JT_1J^*	,			& \nu_1 &\to \nu_1J^* ,	& \nu_1^* &\to (J^*)^{-1}\nu_1^*,\\
T_2 &\to (J^*)^{-1}T_2J^{-1},	& \nu_2 &\to \nu_2J^{-1},	& \nu_2^* &\to J\nu_2^*
\end{align*}
with 
\begin{align}\label{J}
J = R_2^{1/2}R^{-1/2},\quad R =(R_2^{1/2}R_1R_2^{1/2})^{1/2}
\end{align}
(notice that the Berezinian of such change is 1 according to the Jacobian of the change of $T_1,T_2$ and to (\ref{Gr_change})). 
Then
\begin{align*}
&\Tr R_1R_2=\Tr R_2^{1/2}R_1R_2^{1/2}=\Tr R^2,\\
&\Tr R_1JT_1J^*=\Tr R^{-1/2}R_2^{1/2}R_1R_2^{1/2}R^{-1/2} T_1=\Tr R^{-1/2}R^2R^{-1/2} T_1=\Tr R T_1,\\
&\Tr R_2(J^*)^{-1}T_2J^{-1}=\Tr R^{1/2}R_2^{-1/2}R_2R_2^{-1/2}R^{1/2} T_2=\Tr R^{1/2}R^{1/2} T_2=\Tr R T_2,
\end{align*}
and hence we obtain
\begin{equation*}
\begin{split}
&\mathcal{Z}(\zeta,\zeta',\hat\varepsilon, \hat\varepsilon') = \frac{n^{12}}{(2\pi)^8} \int\limits_{\herm_2^+} d R_1 dR_2 \int\limits_{\herm_2} dT_1 dT_2 \int dW 
\int d\nu\\
&\quad{}\times \exp \lbrace\log\, \Sdet Q'+ n \Tr \lbrack - R^2  -iR(T_1 + T_2) - WW^* \rbrack - \Tr (\nu_1\nu_1^* + \nu_2\nu_2^*) \rbrace,
\end{split}
\end{equation*}
where
\begin{align}\notag
&Q' = \begin{pmatrix}
(W - \hat\varepsilon) \otimes I_n & i\mathcal A(\zeta) & n^{-1/2} \nu_1 \otimes I_n & 0 \\
i\mathcal A(\zeta)^*& (W^* - \hat\varepsilon) \otimes I_n & 0 & n^{-1/2} \nu_2 \otimes I_n \\
n^{-1/2} \nu_2^* \otimes I_n & 0 & (iT_1 - J^{-1}\hat\varepsilon'(J^*)^{-1}) \otimes I_n & i\mathcal A(\zeta'_J) \\
0 & n^{-1/2} \nu_1^* \otimes I_n & i\mathcal A(\zeta'_J)^* & (iT_2 - J^*\hat\varepsilon' J) \otimes I_n
\end{pmatrix},\\
&\zeta'_J=J^{-1}\zeta' J
\label{zeta_J}\end{align}
 Here we used that $Q'=\mathcal{J}Q\mathcal{J}^*$ with
\[
\mathcal{J}=\begin{pmatrix}
I_2&0&0&0\\
0&I_2&0&0\\
0&0&J^{-1}&0\\
0&0&0&J^*
\end{pmatrix}\otimes I_n.
\]
The next change is
\begin{align*}
T_1 + iJ^{-1}\hat\varepsilon'(J^*)^{-1}\to T_1, \quad
 T_2 + iJ^*\hat\varepsilon'J\to T_2.
\end{align*}
Notice  that a ``matrix shift'' above means the change of integration contours for corresponding matrix entries. For $T_{1},T_{2}$ it is
possible if the imaginary part of  ``shifting matrix'' is positive, since in this case all eigenvalues of the ``bosonic"  block ($2\times 2$ block at the
 right  bottom corner of matrix $Q_1$ below)  have positive 
real parts and  hence cannot be 0.

Since $\Tr RJ^{-1}\hat\varepsilon'(J^*)^{-1} = \Tr \hat{\varepsilon}'R_1$ and $\Tr R J^*\hat\varepsilon'J = \Tr \hat{\varepsilon}'R_2$, we get
\begin{align*}
&\mathcal{Z}(\zeta,\zeta',\hat\varepsilon, \hat\varepsilon') =\frac{n^{12}}{(2\pi)^8} \int\limits_{\herm_2^+} d R_1 dR_2  \int\limits_{\substack{\Im T_1 = J^{-1}\hat\varepsilon'(J^*)^{-1} \\ \Im T_2 = J^*\hat\varepsilon'J}} dT_1 dT_2 \int dW d\nu\,\\
&\times \exp \lbrace \log\, \Sdet Q_1-n \Tr \lbrack  R^2 + \hat\varepsilon' R_1 + \hat\varepsilon' R_2 -iR(T_1 + T_2) + WW^* \rbrack - \Tr (\nu_1\nu_1^* + \nu_2\nu_2^*) \rbrace,
\end{align*}
where
\begin{equation*}
Q_1 = \begin{pmatrix}
(W - \hat\varepsilon) \otimes I_n &i \mathcal A (\zeta) & n^{-1/2} \nu_1 \otimes I_n & 0 \\
i\mathcal A (\zeta)^* & (W^* - \hat\varepsilon) \otimes I_n & 0 & n^{-1/2} \nu_2 \otimes I_n \\
n^{-1/2} \nu_2^* \otimes I_n & 0 & iT_1 \otimes I_n &  i\mathcal A (\zeta'_J)\\
0 & n^{-1/2} \nu_1^* \otimes I_n & i\mathcal A (\zeta'_J)^* & iT_2 \otimes I_n
\end{pmatrix}
\end{equation*}
Let us shift the domain of  integration with respect to $T_1$, $T_2$ to $\Im T_1 = \Im T_2 = u_*I$. 
Then
\begin{align*}
&\mathcal{Z}(\zeta,\zeta',\hat\varepsilon, \hat\varepsilon') =\frac{n^{12}}{(2\pi)^8} \int\limits_{\herm_2^+} d R_1 dR_2  \int\limits_{\Im T_1 = \Im T_2 = u_*} dT_1 dT_2 \int dW d\nu \,\\
&\times \exp \lbrace \log\, \Sdet Q_1-n \Tr \lbrack  R^2 + \hat\varepsilon' R_1 + \hat\varepsilon' R_2 -iR(T_1 + T_2) + WW^* \rbrack - \Tr (\nu_1\nu_1^* + \nu_2\nu_2^*) \rbrace.
\end{align*}
Now we are going to change variables $(R_1,R_2)$ to $(R, R_2)$. To this end we first change of the variables $(R_1,R_2)$ to $(R_3,R_2)$ with $R_1= R_2^{-1/2}R_3R_2^{-1/2}$ (with a Jacobian $\det\nolimits^{-2} R_2$), and then change  $R_3 = R^2$ (with a Jacobian $4(\Tr R)^2 \det R$, see Lemma \ref{A = B^2}). We obtain
\begin{align*}\notag
&\mathcal{Z}(\zeta,\zeta',\hat\varepsilon, \hat\varepsilon') = \frac{4n^{12}}{(2\pi)^8}\int\limits_{\herm_2^+} \frac{(\Tr R)^2 \det R}{ \det\nolimits^{2} R_2} dRdR_2
 \int\limits_{\Im T_1 = \Im T_2 = u_*} dT_1 dT_2 \int dW \int d\nu\\
&\times  \exp \lbrace - \Tr (\nu_1\nu_1^* + \nu_2\nu_2^*) + \log\, \Sdet Q_1\rbrace\\ \notag
&\times \exp\lbrace-n \Tr \lbrack R^2 +\hat\varepsilon' R_2^{-1/2}R^2R_2^{-1/2} + \hat\varepsilon' R_2 +iR(T_1 + T_2) +WW^* \rbrack\rbrace.
\end{align*}
Finally, we make the shift $W-\hat\varepsilon\to W$ and then change
\begin{equation*}
W=\Lambda U, 
\end{equation*}
where $\Lambda\in \herm_2^+$, $U\in U(2)$. According to Lemma \ref{W = Lambda U}, this change gives
\[
dW\to 
2\pi^3\,(\Tr \Lambda)^2 \det \Lambda  \,d\Lambda \,dU,
\] 
where $d\Lambda$ is the standard integral over $\herm_2^+$ and $dU$ is the Haar measure over $U(2)$. 
If we  change also
\begin{align*}
(T_1,T_2)\to (T,S),\quad T=\frac{1}{2}(T_1+T_2),\quad S=\frac{1}{2}(T_1-T_2),
\end{align*}
then we get  representation (\ref{repr_fin}) for  $\mathcal{Z}(\zeta,\zeta',\hat\epsilon/n, \hat\epsilon'/n)$ with $\hat Q$ replaced by
\begin{equation*}
\tilde  Q = \begin{pmatrix}
\Lambda U \otimes I_n & i\mathcal A(\zeta) &n^{-1/2} \nu_1 \otimes I_n & 0 \\
i\mathcal A(\zeta)^* & U^*\Lambda \otimes I_n & 0 & n^{-1/2}\nu_2 \otimes I_n \\
n^{-1/2}\nu_2^* \otimes I_n & 0 & i(T+S) \otimes I_n & i\mathcal A(\zeta'_J)\\
0 & n^{-1/2}\nu_1^* \otimes I_n & i\mathcal A(\zeta'_J)^*& i(T-S)  \otimes I_n
\end{pmatrix}.
\end{equation*} 
 Now we make one more  transformation of $\tilde Q$.  Define
\[\hat U_1=\mathrm{diag}\{I_2,U,I_2,I_2\}\otimes I_n,\quad \hat U_2=\mathrm{diag}\{U^*,I_2,I_2,I_2\}\otimes I_n.\]
and change 
\[\nu_2\to U\nu_2,\quad \nu_2^*\to \nu_2^* U^*,\quad \hat Q=\hat U_1\tilde Q\hat U_2.\]
Notice that the Berezinian of such change is $1$ according to (\ref{Gr_change}). Therefore, we obtain (\ref{repr_fin}).

$\square$

In Section \ref{s:4} we will also use a modified form of $\hat Q$, which can be obtained  by replacing its second and  third rows  and columns:
\begin{align} \label{F_1,F_2}
 \hat Q=& \begin{pmatrix}F_1\otimes I_n&i\mathcal{Q}\\ i\mathcal{Q}_*&F_2\otimes I_n
\end{pmatrix},\qquad  \qquad F_{1,2}=\begin{pmatrix}
\Lambda &n^{-1/2}\nu_1 \\
n^{-1/2} \nu_2^* &i(T\pm S) 
\end{pmatrix}, \\
\mathcal{Q}:=& I_4\otimes A_z+n^{-1/2}\hat\zeta\otimes I_n,\quad \mathcal{Q}_*:= I_4\otimes A_z^*+n^{-1/2}\hat\zeta_*\otimes I_n,
\notag\\
\label{hat_zeta}
\hat \zeta=&\mathrm{diag}\{\zeta,\zeta'_{J}\},\quad\hat \zeta_*=\mathrm{diag}\{\zeta_U^*,(\zeta'_{J})^*\},\quad\zeta=\mathrm{diag}\{\zeta_1,\zeta_2\}, \quad \zeta'=\mathrm{diag}\{\zeta'_1,\zeta'_2\},
\end{align}
where the last relation means that we represent $\hat Q $ as a block matrix and denote 
$F_1\otimes I_n, i\mathcal{Q}, F_2\otimes I_n,i\mathcal{Q}_*$ corresponding blocks.
It is easy to check that 
\begin{align*}
\Sdet \hat Q=& \Sdet (F_2\otimes I_n)\cdot \Sdet \Big(F_1\otimes I_n+\mathcal{Q} (F_2\otimes I_n)^{-1}\mathcal{Q}_*\Big)\\
=&\Sdet \Big(F_1F_2\otimes I_n+\mathcal{Q} (F_2\otimes I_n)^{-1}\mathcal{Q}_*(F_2\otimes I_n)\Big).
\end{align*}
In addition, 
\begin{align*}
\Str F_1F_2=\Tr\Lambda^2+n^{-1}\Tr\nu_1\nu_1^*+n^{-1}\Tr\nu_2\nu_2^*+\Tr (T^2-S^2).
\end{align*}
Hence $\mathcal{L}(\tilde Q,T,S,R)$ from (\ref{L}) can be also written in the form
\begin{align}\label{L_fin}
\mathcal{L}(\hat Q,T,S,R)=&\log\, \Sdet \Big(F_1F_2\otimes I_n+\mathcal{Q} (F_2\otimes I_n)^{-1}\mathcal{Q}_*(F_2\otimes I_n)\Big)\\ \notag
&-n \,\Str F_1F_2-n\Tr S^2-n\Tr (R+iT)^2
\end{align}
with $F_1\otimes I_n,\mathcal{Q},F_2\otimes I_n,\mathcal{Q}_*$   defined in (\ref{F_1,F_2}) after the above change.

\section{Saddle-point analysis}\label{s:3}

The main goal of the section is to find a saddle-point of $\mathcal{L}$ from (\ref{L}) and to prove that one can restrict the integration in (\ref{repr_fin}) by the
small neighbourhood of the saddle-point.

It is easy to see that $\log\Sdet \hat Q$ from (\ref{repr_fin}) has the form
\begin{align}\notag
\log\Sdet \hat Q=&\log\det Q_1-\log\det Q_2+\mathrm{Grassm}\\
\hbox{where}\quad Q_1=&\left(\begin{array}{cc}\Lambda\otimes I_n&i\mathcal A(\zeta)\\  i\mathcal A(\zeta_U)^*&\Lambda\otimes I_n\end{array}\right),
\quad Q_2=\left(\begin{array}{cc}i(T+S)\otimes I_n&i \mathcal A(\zeta'_J) \\ i\mathcal A^*(\zeta'_J) &i(T-S)\otimes I_n\end{array}\right).
\label{Q_2}
\end{align}
and we denote by $\mathrm{Grassm}$ all terms which
contain Grassmann variables.

Hence,  $\mathcal{L}$ of (\ref{L}) can be rewritten as
\begin{align}
&\mathcal{L}=n(\mathcal{F}_1+\mathcal{F}_2)+\mathrm{Grassm},\label{L_saddle}\end{align}
where
\begin{align}\label{Phi_1}
 &\mathcal{F}_1=n^{-1}\log\det Q_1-\Tr\Lambda^2,\\
&\mathcal{F}_2=-n^{-1}\log \det Q_2-\Tr R^2-2i\Tr RT\label{Phi_2}.
\end{align}
Therefore, we need to study the behaviour of  $\mathcal{F}_1$ and $\mathcal{F}_2$.  
Introduce the function
\begin{align}\label{f}
f(u):=\int\log(u^2+\lambda^2)d\nu_{n,z}(\lambda^2)-u^2.
\end{align}
It is easy to see that $u_*$ defined by (\ref{eq_u}) is its unique maximum point for $u\ge 0$.

The main result of the section is the following proposition:
\begin{proposition}\label{p:saddle} Given any fixed $\tilde M>0$ we have uniformly in $0<\hat\epsilon<\tilde MI_2,0<\epsilon'<\tilde M$,
$|\zeta|< \tilde M,|\zeta'|<\tilde M$
\begin{align}\label{repr_Z}
(\epsilon')^2\mathcal{Z}(\zeta,\zeta',\hat\epsilon/n, \epsilon'I_2/n)=&Cn^4 (\epsilon')^2\int e^{\mathcal{L}(F_1,F_2,\hat\zeta)}E_{*1}(\hat\epsilon,u_*+\Lambda/n^{1/2}, U)\\
&\times E_{*2}(\epsilon'I_2,u_*-iT/n^{1/2}, R_2)\, d\Lambda dTdSd\nu+O(e^{-c\log^2n}),\notag
\end{align}
with  $u_*$ defined by (\ref{eq_u}), Hermitian $2\times 2$ matrices $\Lambda$, $T$, $S$, and  $2\times 2$ matrices $\nu_l$, $\nu_l^*$, $l=1,2$ of independent Grassmann variables with  $d\nu$ of (\ref{dnu}).
Here functions $E_{*1}$, $E_{*2}$ were defined in (\ref{E_*1,2}), and
\begin{align}\label{repr_L}
&\mathcal{L}(F_1,F_2,\hat\zeta)=\Str  \log (F_1F_2\otimes I_n+\mathcal{Q}(F_2\otimes I_n)^{-1}\mathcal{Q}_*(F_2\otimes I_n))-n \,\Str F_1F_2-\Tr S^2
\end{align}
with
\begin{align}\label{F_12}
F_{l}:=&\left(\begin{array}{cc}u_*+n^{-1/2}\Lambda&n^{-1/2}\nu_{l}\\
n^{-1/2}\nu^*_{3-l}&u_*+in^{-1/2}(T\pm S)\end{array}\right),\quad l=1,2,
\end{align}
and $\mathcal{Q}$, $\mathcal{Q}_*$ were defined in (\ref{F_1,F_2}).
\end{proposition}

  \textit{Poof of Proposition \ref{p:saddle}}.  It is easy to see that  to get (\ref{repr_Z}) from (\ref{repr_fin})  we need to prove  that $\Lambda=u_*I_2$,
  $T=iu_*I_2$, $R=-iT$ and $S=0$ is a saddle-point of $\mathcal{L}$ of (\ref{L_saddle}) and we can restrict the integration by the $O(n^{-1/2}\log n)$ neighbourhood
  of the point.
  
  Let us prove first that $\Lambda=u_*I_2$ is a saddle-point  of
$\mathcal{F}_1$ of (\ref{Phi_1}).  Since   $\mathcal{A}(\zeta)=A_z+O(n^{-1/2})$ and $\mathcal{A}^*(\zeta_U)=A_z^*+O(n^{-1/2})$ for any $U$,
we have
\begin{align*}
\mathcal{F}_{1}=\mathcal{F}_{10}+O(n^{-1/2}),\quad \mathcal{F}_{10}:=\int\log\det(\Lambda^2+\lambda^2 I_2)d\nu_{n,z}(\lambda^2)-\Tr\Lambda^2,
\end{align*}
and so one need to  study a saddle-point of $\mathcal{F}_{10}$

It is easy to see that if $\Lambda=V\mathrm{diag}\{u_1,u_2\}V^*$,
 then
\begin{align*}
\mathcal{F}_{10}=f(u_1)+f(u_2)\le 2f(u_*)
\end{align*}
where $f$ was defined in (\ref{f}).
Since $f(u)$ has only one maximum $u=u^*$ for $u\ge 0$, we obtain that $\Lambda=u_*I_2$ is a saddle-point of
$\mathcal{F}_{10}$, and so only $O(n^{-1/2}\log n)$-neighbourhood  of $u_*I_2$ can contribute to our integral. Expanding
$\mathcal{F}_{1}$ around $u_*I_2$ up to the first order, we get
\begin{align}\label{k_A}
\mathcal{F}_{1}\sim  2f(u_*)+n^{-1/2}(k_A\Tr\zeta^*+\bar k_A\Tr\zeta)+O(n^{-1})
\end{align}
with $k_A$ of (\ref{rho}).

Analysis of saddle-points of $\mathcal{F}_2$ is more involved than that of $\mathcal{F}_1$ since the structure of $\mathcal{F}_2$ is more complicated. Another difficulty comes from the fact that for some
$R,R_2$ we have $\|\zeta'_J\|\sim n^{1/2}$, and so we cannot neglect this term in the saddle-point analysis (in contrast to $\zeta_U$ appearing in $\mathcal{F}_1$). Hence, first of all we need to exclude the situation when $\|\zeta'_J\|$ is big, i.e. 
$\|\zeta'_J\|\gg \log n$.

We denote by $\omega$ the set of all integration parameters
in (\ref{repr_fin}) and consider the sets
\begin{align}\label{Omega_j}
\Omega_1=&\{\omega:\|n^{-1/2}R_2^{-1/2}\zeta' R_2^{1/2}\|>n^{-1/2}\log n\vee \epsilon'\Tr R_2^{-1}>\log^2n\},\\
\Omega_2=&\{\omega: \|R\|>M\vee \|R^{-1}\|>\delta_*^{-1}\},\quad \delta_*:=\log^{-1}n, \notag\\
\Omega_3=&\{\omega: \|T-iu_*I_2\|>n^{-1/2}\log n\vee  \|R+iT\|\ge n^{-1/2}\log n\}\cap\Omega_1^c\cap\Omega_2^c,\notag\\
\Omega_4=&\{\omega:\|S\|>n^{-1/2}\log n\}\cap\Omega_1^c\cap\Omega_2^c\cap\Omega_3^c,\notag
\end{align}
where we denote by $\Omega_j^c$ the complement of $\Omega_j$.

The assertion of Proposition \ref{p:saddle} will follow from the bounds
\begin{align}\label{b_Omega}
\epsilon'^2 e^{\check{\mathcal{L}}}\langle \mathbf{1}_{\Omega_j}\rangle\le e^{-c\log^2 n},\quad j=1,\dots, 4,
\end{align}
if they hold uniformly in  $0<\hat\epsilon<\tilde M,0<\epsilon'<\tilde M, |\zeta|<\tilde M, |\zeta'|<\tilde M$ with any fixed $\tilde M$.
Here we denote by $\langle\phi\rangle$ the integral of the form (\ref{repr_fin}) with a function $\phi$ added as a multiplier before
the exponent. We also set 
\begin{align}\label{check_L}
\check{\mathcal{L}}:=-n^{1/2}(k_A\mathrm{Str}\hat \zeta_*+\bar k_A\mathrm{Str}\hat \zeta)
\end{align}
(see (\ref{rho}) for the definition of $k_A$ and (\ref{hat_zeta}) for the definition of $\hat \zeta$ and $\hat \zeta_*$). This term appears in the expansion of $\mathcal{L}$ 
of (\ref{L_saddle}) near its saddle-point  (see, e.g. (\ref{k_A})). Since at the end of our proof after some differentiation procedure we put  $\zeta=\zeta'$,
and then get $\check{\mathcal{L}}=0$, this term is not important for us (see the discussion after formula (\ref{der_Z}) in Section \ref{s:4}).

The multiplier $\epsilon'^2$ in (\ref{b_Omega})   appears  because of $E_{*2}$ in (\ref{repr_fin}). We need to control  the  dependence of 
the bounds 
on $\epsilon'$ since in Section \ref{s:5} we need to integrate over $\epsilon'$ from $\epsilon'=0$. 

Below we use also that  if $\Re R^2>\delta_0$, then
\begin{align}\label{int_R_2}
&\Big|\int \frac{dR_2}{(\det R_2)^2}\exp\{-\epsilon'\Tr R^2R_2^{-1}-\epsilon'\Tr R_2\}\Big|\\
\le& C\int_0^\infty\int_0^\infty  d\rho_1d\rho_2(\rho_1^{-1}-\rho_2^{-1})^2\exp\{-\epsilon'(\delta_0(\rho_1^{-1}+\rho_2^{-1})+\rho_1+\rho_2)\}\le C'(\delta_0\epsilon'^{2})^{-1}.
\notag\end{align}
Here we change  $dR_2$ over the positive Hermitian matrices to the integration with respect to eigenvalues $\rho_1,\rho_2$ of $R_2$. Then Jacobian $\tfrac{\pi}{2}(\rho_1-\rho_2)^2$ divided 
by $(\det R_2)^2$ gives the multiplier $(\rho_1^{-1}-\rho_2^{-1})^2$.

Integrating first with respect to $R$ and then with respect to $R_2$, we get  the bound
\begin{align}\label{int_R_R_2}
&\int\det R\,(\Tr R)^2 dR\int \frac{dR_2}{(\det R_2)^2}\exp\{-\epsilon'\Tr R^2R_2^{-1}-\epsilon'\Tr R_2-n\Tr R^2\}\\
=& C\int_0^\infty\int_0^\infty \frac{(\rho_1-\rho_2)^2d\rho_1d\rho_2}{(n+\epsilon'\rho_1^{-1})^2(n+\epsilon'\rho_1^{-1})^2\rho_1^2\rho_2^2}
\exp\{-\epsilon'(\rho_1+\rho_2)\}\notag\\
=& C\int_0^\infty\int_0^\infty\frac{(\tilde\rho_1-\tilde\rho_2)^2d\rho_1d\rho_2}{(n\tilde \rho_1+\epsilon'^{2})^2(n\tilde \rho_2+\epsilon'^{2})^2}
\exp\{-(\tilde\rho_1+\tilde\rho_2)\}\le C(n\epsilon')^{-2}.
\notag\end{align}
Bound (\ref{b_Omega}) for $\Omega_1$ follows from the lemma
\begin{lemma}\label{l:b_R2} Denote by $I(R_2)$ the integral which we obtain if in (\ref{repr_fin}) fix $R_2$ and integrate with respect
to the rest of parameters.
Then there are some fixed positive $p_1,p_2$ such that
\begin{align}
|I(R_2)e^{\check{\mathcal{L}}}|\le &\frac{Cn^{p_1}}{(\det R_2)^{p_2}}\exp\Big\{-n\log\Big(1+\|n^{-1/2}R_2^{-1/2}\hat\zeta R_2^{1/2}\|^2/2u_*^2\Big)-\epsilon'\Tr(R_2+u_*^2R_2^{-1})\Big\}
\label{b_R2.1}\end{align}
where $\check{\mathcal{L}}$ is defined in (\ref{check_L}).
\end{lemma}
The proof of the lemma is given after the proof of Proposition \ref{p:saddle}.

Note that if we are in $\Omega_1^c\cap\Omega_2^c$, then
\[
\|n^{-1/2}\zeta_J\|\le M^{1/2}\delta_*^{-1/2}\|n^{-1/2}R_2^{-1/2}\hat\zeta R_2^{1/2}\|\le Cn^{-1/2}\log^2 n,
\]
and we can consider saddle-points of $\mathcal{F}_2$ of (\ref{F_1,F_2}) with $\zeta=0$  only, as we did for $\mathcal{F}_1$.

 Now let us prove (\ref{b_Omega}) for $\Omega_2$.  For $\|R\|>M$ the  bound is evident, since we have
the term $-n\Tr R^2$ at the exponent (so it is sufficient to move the integration with respect to $T$ to $T=i+T'$). 
Suppose now  $\|R^{-1}\|>\delta_*^{-1}$ which means that $r_2$,
the minimum eigenvalue of $R_2$, satisfies the bound   $r_2< \delta_*$.  Move the integration with respect to $T$ such that 
in the  basis of eigenvectors of $R>0$ we have $T_{11}=i+\mathbb{R}$ and $T_{22}\in i\delta_*^{-1}+\mathbb{R}$, $T_{12}=\bar T_{21}$.
Then
\begin{align*}
\Re\mathcal{F}_2&=-n^{-1}\Re  \log \det Q_2-\Re(2i\Tr TR+\Tr R^2)\\ &\le -\log \delta_*^{-1}+2r_1+2-r_1^2-r_2^2
\le -\log \delta_*^{-1}+8-\Tr R^2/2,
\end{align*}
and  (\ref{int_R_R_2}) (with $\Tr R^2$ replaced by $\Tr R^2/2$) gives us the uniform with respect to $\epsilon'$ bound of the type (\ref{b_Omega}) for $\Omega_2$.

\medskip

To study the contribution of $S$ to $\mathcal{F}_2$  we need one more lemma. Set
 \begin{align}\label{Omega}
 \tilde\Omega=\{t:\pi/4<\mathrm{arg} \,t<3\pi/4\}.
\end{align}  

\begin{lemma}\label{l:1}Set
\begin{align}\label{M(T,S)}
\mathcal{M}(T,S,\lambda  I_2)=\left(\begin{array}{cc}T+S&\lambda I_2\\ \lambda I_2& T-S\end{array}\right).
\end{align}
Then for any $t_1,t_2\in \partial \tilde\Omega$ with $\tilde\Omega$ of (\ref{Omega}), $T=\mathrm{diag}\{t_1,t_2\}$, $\lambda>0$  and  $S=S^*$ we have
\begin{align}\label{l1.1}
\Re \log \det \mathcal{M}(T,S,\lambda I_2)\ge \Re \log \det \mathcal{M}(T,0, \lambda  I_2).
\end{align}
\end{lemma}
\begin{corollary}\label{cor:1}
By the maximum principle  inequality (\ref{l1.1}) is valid for all $t_1, t_2\in  \tilde\Omega$.
\end{corollary}
It follows from the corollary that for  any normal $T$ with eigenvalues $t_1,t_2\in  \tilde\Omega$
the point $S=0$ is a maximum point for $\Re \mathcal{F}_2(T,S)$, and hence the integration over $S$ can be restricted to a small neighbourhood 
of $S=0$. In particular, it implies (\ref{b_Omega}) for $\Omega_4$.

\medskip

Let us consider $R$ in the basis of eigenvectors of $T$. The complement of  $\Omega_2$ in the set of positive matrices are the matrices of the form
\begin{align*}
 R=\delta_* I_2+\tilde R, \quad \tilde R>0.
\end{align*}
For simplicity, we will omit the tilde below, i.e. change the $R$ integration $R\to \delta_* I_2+R$ with
\[
R_{11}> 0,\, R_{22}>0,\, R_{12}=\bar R_{21},\quad | R_{12}|^2\le R_{11}R_{22}.
\]
Furthermore, the integration with respect to $R$ can be replaced by the integration by $R_{11}$, $R_{22}$, $R_{12}=e^{i\phi}\theta (R_{11}R_{22})^{1/2}$ ($0<\theta<1$,
$0\le \phi<2\pi$). Then we have at the exponent 
\begin{align}\label{exp}
-n(R_{11}^2+R_{22}^2+2\theta^2 R_{11}R_{22})-2in\Tr RT-n(2\delta_*(\Tr R+i\Tr T)+ n^{-1/2}\varphi(R))
\end{align}
with some analytic function $\varphi(R)$. 

For $t_1,t_2$ which are eigenvalues of $T$ we set the integration contour  as follows
 \begin{align}\label{cont_T}
\mathcal{L}=&\mathcal{L}_0\cup \mathcal{L}_+\cup \mathcal{L}_-,\quad \mathcal{L}_0=\{t(\tau)=(-u_*^2+i\tau)^{1/2},\Im t(\tau)<M\},\quad \\
\mathcal{L}_{+}=&\{t_+(M)+\tau,\tau>0\}, \quad \mathcal{L}_{-}=\{t_-(M)-\tau,\tau>0\},
\notag \end{align}
 where  we have chosen a branch of the root such that $t(0)=iu_*$, and denote by
 $t_{+}(M),t_{-}(M)$ the points of intersection of $t(\tau)$ with the line $\Im t=M$. 
 Note that 
 \begin{align}\label{b_t^2}
 -\Re\,t^2\le 2u_*^2,\quad t\in \mathcal{L}\cap\tilde\Omega.
\end{align}  
Indeed, for $t\in \mathcal{L}_0$ we have $ -\Re\,t^2=u_*^2$ by the definition, and for $t=t_++\tau$, $0<\tau<M-\Re t_+$
\begin{align*}
|\Re t_+-M|\le \frac{u_*^2}{2M}+O(M^{-2})\quad\Rightarrow\quad -\Re\,t^2\le 2u_*^2.
 \end{align*}
Now for any $t_1,t_2\in L$ choose the contour of integration with respect to $R_{11}$ and $R_{22}$ 
 \begin{align}\label{cont_R}
R_{11}\in \mathcal{L}(t_1),\, R_{22}\in \mathcal{L}(t_2), \,\,\mathcal{L}(t)=
[0,-it-\delta_*]\cup[-it-\delta_*,+\infty),
 \end{align}
 We remark here that for any $t$ the contribution of the integral with respect to \\ $R_{jj}\in [-it-\delta_*+n^{-1/2}\log n,\infty]$ will be
 $O(e^{-c\log^2n})$, hence we do not consider it in details. It means that in all bounds below it is sufficient to consider $R_{jj}=(-it_j-\delta_*)\tau_j$,
 $\tau_j\in[0,1]$, and hence we can write
 \begin{align}\label{int_R}
\Re\{-n(R_{jj}+\delta_*)^2-2in(R_{jj}+\delta_*)t_j\}&=-n(1-\tau_j)^2\Re(-\delta_*-it_j)^2-n\Re t_j^2\\ &=
 -n\Re t_j^2(1-(1-\tau_j)^2)+O(n\delta_*).
\notag \end{align}
 For $t_1,t_2\in \mathcal{L}_+\not\in\tilde\Omega$ ($t_{1,2}=M+\delta_{1,2}+iM$, $\delta_1,\delta_2>0$)  we have the bound 
  \begin{align}\notag
n\Re\mathcal{F}_2\le&\Re\{- \log \det Q(T,S)+n K(\tau_1,\tau_2)+O(n\delta_*)\}, \\
 K(\tau_1,\tau_2,\theta)=&(\tau_1^2-2\tau_1)A_1+(\tau_2^2-2\tau_2)A_2+2\theta^2\tau_1\tau_2 B,\label{ReK}\\
  & A_1=\Re t_1^2,\quad
A_2=\Re t_2^2,\quad B=\Re t_1t_2.\notag
\end{align}
In this case
\begin{align*}
& A_1=\Re t_1^2=(M+\delta_1)^2-M^2,\quad
A_2=\Re t_2^2=(M+\delta_2)^2-M^2,\\
&B=\Re t_1t_2=(M+\delta_1)(M+\delta_2)-M^2,\notag
\end{align*}
where $\delta_1,\delta_2$ are some positive numbers. Let us check that 
 \begin{align}\label{ReK<0}
  K(\tau_1,\tau_2,\theta)\le 0.
 \end{align}
If one of $\tau_1,\tau_2$ is zero, then (\ref{ReK<0}) is evident, since $A_1,A_2$ are positive and $\tau_i^2-2\tau_i\le 0$. 
Since for $t_1,t_2\in \mathcal{L}_+\not\in\tilde\Omega$ we have $B>0$,  for $\tau_1,\tau_2>0$ the maximum in $\theta$ could be obtained only with $\theta=1$. 
Note also that for fixed $\tau_2$  $ K(\tau_1,\tau_2,1)$ is quadratic with respect to
  $\tau_1$ with a positive coefficient $A_1$ in front of $\tau_1^2$, and hence the maximum in $\tau_1$ can be achieved only at $\tau_1=0$ or $1$. By the same reason
   $\tau_2=0$ or $1$. Finally,
\[
 K(1,1,1)=-(\delta_1-\delta_2)^2\le 0,
\]
which finishes the proof of (\ref{ReK<0}). It implies
\begin{align*}
n\Re\mathcal{F}_2\le -2n\log M.
\end{align*}
Remark also that in this case we do not have the bound $\Re R^2>\delta_0$ needed for (\ref{int_R_2}), but, since we are in $\Omega_1^c$,
the integral of $E_2$ is bounded  by
\begin{align*}
& \int_{\rho_1\ge\rho_2>\epsilon'/\log^2n} d\rho_1d\rho_2(\rho_1^{-1}-\rho_2^{-1})^2\exp\{-\epsilon'(\rho_1+\rho_2)\}
 \le \frac{C}{\epsilon'}\int_{\epsilon'/\log^2n}^\infty\frac{d\rho_2}{\rho_2^2} \le C\frac{\log^2n}{\epsilon'^{2}}.
\end{align*}
Taking into account the above bound for $\Re\mathcal{F}_2$, this gives (\ref{b_Omega}).

The cases $t_1,t_2\in \mathcal{L}_-$ and $t_1\in \mathcal{L}_+,\,t_2\in \mathcal{L}_-$ ($t_1,t_2\not\in\tilde\Omega$) can be analysed similarly.

When $t_1\not\in \tilde\Omega,\,t_2\in \tilde\Omega$ ($t_1=M+\delta_1+iM$) we have
 \begin{align*}
n\Re\mathcal{F}_2\le&\Re\{- \log Q(T,S)+n K(\tau_1,\tau_2)-n t_2^2+O(n\delta_*)\} \\
 K(\tau_1,\tau_2,\theta)=&(\tau_1^2-2\tau_1)A_1+(1-\tau_2)^2A_2+2\theta^2\tau_1\tau_2B.
 \end{align*}
Again we want to check (\ref{ReK<0}). Since $A_1>0$, $A_2<0$, the non-trivial case is $B>0$. Then again one should consider $\theta=1$, since $B>0$, and  $\tau_1=0$ or $1$,  
  since $A_1>0$. If $\tau_1=0$, $\Re K(\tau_1,\tau_2,1)\le 0$. If $\tau_1=1$, it is easy to see that the maximum point in $\tau_2$ is $\tau_2=1$.  But
   setting $t_2=x+iy$ and using that $x< y\le M$, for $t_2\in \tilde\Omega$, we get
 \begin{align*}
K(1,1,1)=&-A_1+2B=-(2M\delta_1+\delta_1^2)+2((M+\delta_1)x-My)\\
 \le& -(2M\delta_1+\delta_1^2)+2\delta_1x\le -(2M\delta_1+\delta_1^2)+2\delta_1M,
 =-\delta_1^2
 \end{align*}
which finishes the proof of (\ref{ReK<0}). Then, using (\ref{b_t^2}), we obtain
\begin{align*}
n\Re\mathcal{F}_2\le&\Re\{\max_{t_{1}\in \mathcal{L}_+,t_2\in \mathcal{L}_0\cap\Omega}\{-\Re\log\det Q(T,S)-n\Re t_2^2\}\}\le \exp\{-n(\log M-2u_*^2)\}.
\end{align*}

We are left to consider the case when $t_1,t_2\in \tilde\Omega$. Since in this case $\Re \,t_1t_2<0$, one should consider only the case $\theta=0$.
 Using Lemma \ref{l:1}, we obtain
 \begin{align}
 n\Re\mathcal{F}_2\le &\Re\{-\log \det Q(T,S)-2ni\Tr TR-\Tr R^{2}+n^{1/2}\varphi(R)+O(n\delta_*)\} \notag\\
 =&\Re\{-\log \det Q(T,S)-n\Tr T^2+O(n^{1/2})+O(n\delta_* )\}\notag\\
 &\le \Re\{-\log \det Q(T,0)-n\Tr T^2+O(n^{1/2})+O(n\delta_* )\}\notag\\
 =&\Re\{n(f_1(t_1)+f_1(t_2))+O(n^{1/2})+O(n\delta_*)\},\label{in_F2} \\
 f_1(t)=&-\int\log(\lambda^2-t^2)d\nu_z(\lambda^2)-t^2.
\notag \end{align}
 Since $-t^2=u_*^2-i\tau$ ($t\in \mathcal{L}_0$),  we have evidently 
  \[\Re(f_1(t_1)+f_1(t_2))\le -2f(u_*),\]
 where $f$ was defined in (\ref{f}). Observe also that for $t_1,t_2$  sufficiently close to $iu_*$
 we can use  the first line of (\ref{int_R}), and since the first term in the r.h.s. here becomes negative for $\tau_j<1$,
 we get
  \begin{align}\label{in_F2a}
 n\Re\mathcal{F}_2\le -2nf(u_*)+O(n^{1/2}),
\end{align}
  for all $t_1,t_2,\tau_1,\tau_2,\theta $ belonging to our contours.  In addition, for $T=iu_*+T'$ with $T'=T'^*$, $\|T'\|\sim n^{-1/2}\log n$,
   expanding the functions near  $T=iu_*I_2$, we obtain
 \begin{align*}
 -\Re\log\det Q(iu_*+T',0)-n\Tr (iu_*+T')^2\le -2nf(u_*) -n c_2\Tr T'^2\le -2nf(u_*) -c_2\log^2n
 \end{align*}
 with $c_2=-f''(u_*)>0$. Besides,
 $O(n^{1/2})$  term in (\ref{in_F2a})  becomes $n^{1/2}(-k_A\Tr\zeta'^*-\bar k_A\Tr\zeta')$. Hence, using (\ref{k_A}),
 we obtain (\ref{b_Omega}) for $\Omega_3$. This  completes the proof of Proposition \ref{p:saddle}.

$\square$

\textit{Proof of Lemma \ref{l:b_R2}.}
Let us come back to our construction of the integral representation (\ref{Z_av}) and apply the matrix Fourier transform for the function
 (\ref{cond_X}) with $j=2$. We have
\begin{align*}
I(R_2)=
&Cn^{12}\int dXd\Psi dT_2\, Z(X)\tilde Z(\Psi)\\
&\exp\Big\{-n^{-1}\Tr X^{(1)}R_2X^{(1)*}+i\Tr (X^{(2)*}X^{(2)}-nR_2)T_2\\
&-n^{-1}\Big(\Tr \Psi^{(1)*}\Psi^{(1)}\Psi^{(2)*}\Psi^{(2)} +\Tr \Psi^{(1)*}X^{(1)}X^{(2)*}\Psi^{(2)}+\Tr \Psi^{(2)*} X^{(2)}X^{(1)*}\Psi^{(1)}\Big) \Big\},
\end{align*}
where $Z(X)$ and $\tilde Z(\Psi)$ were defined in (\ref{ZZE}).
Applying (\ref{Hubb-Str}) to the terms in the last line and integrating with respect to $d\Psi$,
we get
\begin{align*}
I(R_2)=
&Cn^{12}\int dXdT_2 d\Lambda dU d\nu  \, Z(X)E_{*1}(\hat\epsilon,\Lambda,U)\\
&\exp\Big\{-n^{-1}\Tr X^{(1)}R_2X^{(1)*}-(\epsilon'/n)\Tr X^{(2)*}X^{(2)}+i\Tr (X^{(2)*}X^{(2)}-nR_2)T_2\\
&+n\mathcal{F}_1(\Lambda)+n^{-1}\sum_{j,k=1}^2\Tr G(\Lambda) X^{(j)}\nu^{(j)}\nu^{(k)*}X^{(k)*}-\Tr(\nu_1\nu_1^*+\nu_2\nu_2^*)
 \Bigr\rbrace,
\end{align*}
where $\mathcal{F}_1$ is defined in (\ref{Phi_1}), $E_{*1}$ is defined in (\ref{E_*1,2}), and for $\Lambda=U\mathrm{diag}\{u_1,u_2\}U^*$ we set $G(\Lambda)=
U\mathrm{diag}\{(A^*_zA_z+u_1^2)^{-1}, (A^*_zA_z+u_2^2)^{-1}\} U^*$. Since the sum at the exponent is a quadratic form of a finite number 
of Grassman variables, we can write
\begin{align*}
\exp\Big\{n^{-1}\sum_{j,k=1}^2\Tr G(\Lambda) X^{(j)}\nu^{(j)}\nu^{(k)*}X^{(k)*}\Big\}=
\sum_{m=0}^p\frac{1}{m!}\Big(n^{-1}\sum_{j,k=1}^2\Tr G(\Lambda) X^{(j)}\nu^{(j)}\nu^{(k)*}X^{(k)*}\Big)^m
\end{align*}
with some fixed $p$.  Using a saddle-point analysis with respect to $\mathcal{F}_1$, we conclude that $G(\Lambda)$ above  can be replaced
by $I_2\otimes(A^*_zA_z+u_*^2)^{-1}+o(1)$ and $\mathcal{F}_1(\Lambda)$ by its expansion (\ref{k_A}) near the saddle-point.
Then, integrating first with respect to $X^{(1)},X^{(1)*}$ and  then with respect to $X^{(2)},X^{(2)*}$, we get
\begin{align*}
I(R_2)&=e^{2nf(u_*)+n^{1/2}(k_A\Tr\zeta^*+\bar k_A\Tr\zeta)}\tilde I(R_2),\\
\tilde I(R_2)&=Cn^{p_1}(\det R_\epsilon)^{-n}\int \mathcal{P}(X^{(2)},X^{(2)*}, R_\epsilon )dX^{(2)}dT_2E_{*1}(\hat\epsilon,\Lambda,U)\\
&\exp\{-\Tr\mathcal{A}^* (\zeta)\mathcal{A} (\zeta)X^{(2)}R_{\epsilon}^{-1}X^{(2)*}-(\epsilon'/n)\Tr X^{(2)}X^{(2)*}+i\Tr (X^{(2)*}X^{(2)}-nR_2)T_2\},\\
&=\frac{n^{p_1}}{(\det R_\epsilon)^{n}}\int dT_2\tilde{\mathcal{P}}(T_2,R_\epsilon)\\
&\qquad\times\exp\Big\{-\log\det\Big((-iT_2+\epsilon/n)\otimes I_n+\mathcal{A}^* (\zeta')(R^{-1}_\epsilon\otimes I_n)\mathcal{A} (\zeta')\Big)-in\Tr R_2T_2\Big\}. 
\end{align*} 
Here $k_A$ is defined by (\ref{rho}),  $f$ is defined by (\ref{f}), 
$R_\epsilon=R_2+\epsilon'/n$,
 $ \mathcal{P}(X^{(2)},X^{(2)*}, R_\epsilon )$ is some polynomial of $(X^{(2)},X^{(2)*})$,  and  $\tilde{\mathcal{P}}(T_2,R_2)$
 is a result of the application of the Wick theorem to $\mathcal{P}(X^{(2)},X^{(2)*}, R_\epsilon )$. Here $ \mathcal{P}(X^{(2)},X^{(2)*}, R_\epsilon )$
 and $\tilde{\mathcal{P}}(T_2,R_\epsilon)$ satisfies the bounds
\begin{align*}
& | \mathcal{P}(X^{(2)},X^{(2)*}, R_\epsilon )|\le C(\det R_\epsilon)^{-q}(n^{-1}\Tr X^{(2)*}X^{(2)})^p,\\
& |\mathcal{P}(X^{(2)},X^{(2)*}, R_\epsilon )|\le C(\det R_\epsilon)^{-q} \|\big((-iT_2+\epsilon/n)\otimes I_n+\mathcal{A}^* (\zeta)R^{-1}_\epsilon\otimes I_n\mathcal{A} (\zeta)\big)^{-1}\|^{2p}.
\end{align*}  
Changing the variables
\begin{align}\label{chT_2}
T_2+i\epsilon'/n=R_\epsilon^{-1/2}\tilde T R_\epsilon^{-1/2},
\end{align}
we get
\begin{align*}
\tilde I(R_2)= C&n^p(\det R_\epsilon)^{-q}\int \hat{\mathcal{P}}(\tilde T,R_\epsilon)d\tilde T\exp\{-\log\det(-i\tilde T\otimes I_n+\mathcal{A}^* (\zeta_R')\mathcal{A} (\zeta_R'))\\
&-in\Tr \tilde T_2-\epsilon'\Tr R_2+i\epsilon'\Tr\tilde TR_\epsilon^{-1}\},
\end{align*} 
where $ \zeta_R'=R_\epsilon^{-1/2} \zeta' R_\epsilon^{1/2}$, and $\hat{\mathcal{P}}(\tilde T,R_\epsilon)$ is a result of the change (\ref{chT_2}) in $\tilde{\mathcal{P}}(T_2,R_\epsilon)$.  Move the integration over $\tilde T$ to $iu_*^2+T'$ with $T'=(T')^*$.  Then evidently
\begin{align}\notag
\Re\{-\log\det(u_*^2-i\tilde T'\otimes I_n+&\mathcal{A}^* (\zeta_R')\mathcal{A} (\zeta_R'))+2nu_*^2-in\Tr T'-\epsilon'\Tr R_2-\epsilon'\Tr(u_*^2 -iT')R_\epsilon^{-1}\}\\
&\le -\log\det(u_*^2+\mathcal{A}^* (\zeta_R')\mathcal{A} (\zeta_R'))+2nu_*^2-\epsilon'\Tr (R_2+ u_*^2R_\epsilon^{-1})\notag \\
&-\frac{1}{2}\log\det (1+C_\zeta^2(T')^2\otimes I_n),
\label{b_Re}\end{align}
where $C_\zeta=(u_*^2+(\|A_z\|+\|n^{-1/2}\zeta_R'\|)^2)^{-1}$, and for $B=u_*^2+\mathcal{A}^* (\zeta_R')\mathcal{A} (\zeta_R')$ we  used the bound 
\begin{align*}
&\Re\{-2\log\det(B+iM)\}=-2\log\det B-\log\det (1+B^{-1/2}MB^{-1}MB^{-1/2})\\
&\le-2\log\det B-\log\det (1+C^2M^2)
\end{align*}
valid for any matrices $B>0$,  $M=M^*$, if $B^{-1}\ge C$. 

In addition, denoting $G= (u_*^2+A_z^*A_z)^{-1}$, we get
\begin{align*}
-\log\det(u_*^2+&\mathcal{A}^* (\zeta_R')\mathcal{A} (\zeta_R'))=-\log\det(u_*^2+I_2\otimes A_z^* A_z)\\
&-\log\det(1+ n^{-1/2}\zeta_R'^{*}\otimes G^{1/2}A _zG^{1/2}+n^{-1/2}\zeta_R'\otimes G^{1/2}A_z^*G^{1/2}+n^{-1}\zeta_R'^{*}\zeta_R'\otimes G).
\end{align*}
Continuing transformations, we obtain
\begin{align*}
&\tilde{\mathcal{D}}:=\log\det(1+n^{-1/2}\zeta_R'^*\otimes G^{1/2}A_z G^{1/2}+n^{-1/2}\zeta_R'\otimes G^{1/2}A_z^*G^{1/2}+n^{-1}\zeta_R'^*\zeta_R'\otimes G)\\
&=\log\det\Big((1+n^{-1/2}\zeta_R'^*\otimes G^{1/2}A_z G^{1/2})(1+n^{-1/2}\zeta_R'\otimes G^{1/2}A_z ^*G^{1/2})\\
&\hskip2cm +n^{-1}\zeta_R'^*\zeta_R'\otimes G^{1/2}(1-A_zGA^*_z)G^{1/2}\Big).
\end{align*}
To simplify formulas below, denote
\begin{align*}
K=&G^{1/2}A _z^*G^{1/2},\quad K_0=G^{1/2}(1-A_zGA_z^*)G^{1/2}=u_*^2G^{1/2}\tilde GG^{1/2},\quad\tilde G=(u_*^2+A_zA_z^*)^{-1}\\
Y=&1+n^{-1/2}\zeta'\otimes K\Rightarrow (1+n^{-1/2}\zeta_R'^*\otimes K)=(R_\epsilon^{-1/2} \otimes I_n)Y( R_\epsilon^{1/2}\otimes I_n).
\end{align*}
Then
\begin{align*}
\tilde{\mathcal{D}}=&\log\det Y+\log\det Y^*\\&
+\log\det(1+n^{-1}(R_\epsilon^{1/2}\otimes I_n) Y^{-1}(\zeta 'R_\epsilon^{-1}\zeta'^*\otimes K_0)(Y^{*})^{-1}(R_\epsilon^{1/2}\otimes I_n)).
\end{align*}
Since $n^{-1}\Tr K=k_A$, and $\|K\|\le C$ ,we get
\begin{align*}
&\log\det Y+\det Y^*=n^{1/2}(k_A\Tr\zeta'^*+\bar k_A\Tr\zeta')+O(1),\quad\\
&Y^{-1} =1+O(n^{-1/2}).
\end{align*}
Hence, using the bound $K_0\ge u_*^{-2}$, 
\begin{align*}
-\tilde{\mathcal{D}}=&-n^{1/2}(k_A\Tr\zeta'^*+\bar k_A\Tr\zeta')+O(1)-
\log\det(1+n^{-1}\zeta_R^*\zeta_R\otimes K_0+O(n^{-1/2}\|n^{-1/2}\zeta_R\|^2)\\
\le &-n^{1/2}(k_A\Tr\zeta'^*+\bar k_A\Tr\zeta')-O(1)-n\log\det(1+\|n^{-1/2}\zeta_R\|^2/2u_*^2).
\end{align*}
The   above bounds  imply
\begin{align*}
& -\log\det(u_*^2+\mathcal{A}^* (\zeta_R')\mathcal{A} (\zeta_R'))+2nu_*^2\\
 &\qquad \le -2nf(u_*)-n^{1/2}(k_A\Tr\zeta'^*+\bar k_A\Tr\zeta')
 -n\log\det(1+\|n^{-1/2}\zeta_R\|^2/2u_*^2),
\end{align*}
and this inequality combined with bound (\ref{b_Re}) yields (\ref{b_R2.1}).

$\square$

\textit{Proof of Lemma \ref{l:1}.\,}
Consider the case when $\mathrm{arg}\,t_1=\mathrm{arg}\,t_2=\pi/4$, $T=e^{i\pi/4}\mathcal{T}$, $\mathcal{T}> 0$.
Then
\begin{align*}
\det \mathcal{M}(T,S,\lambda I_2)=&\det\mathcal{T}^2\det\mathcal{M}(e^{i\pi/4},S_T,D)=
\det\mathcal{T}^2\det \Big(e^{i\pi/4}I_4+\mathcal{M}(0,S_T,D)\Big)\\
 S_T=&\mathcal{T}^{-1/2} S \mathcal{T}^{-1/2},\quad D=\lambda \mathcal{T}^{-1}.
\end{align*}
It is easy to see that if $\lambda'$ is an eigenvalue of $\mathcal{M}_0=\mathcal{M}(0,S_T,D)$ with eigenvector $(x,y)$ ($x,y\in\mathbb{C}^2$
then $(-\lambda')$ also  is  an eigenvalue of $\mathcal{M}(0,S_T,D)$ with eigenvector $(y,-x)$.
Hence
\begin{align*}
\Big|\det \Big(e^{i\pi/4}I_4+\mathcal{M}(0,S_T,D)\Big)\Big|=&\Big(\lambda_1+e^{i\pi/4})(-\lambda_1+e^{i\pi/4})
(\lambda_2+e^{i\pi/4})(-\lambda_2+e^{i\pi/4})\Big|\\
=&\Big|(\lambda_1^4+1)(\lambda_2^4+1)\Big|^{1/2}=\Big|1+\frac{1}{2}\Tr \mathcal{M}_0^4+(\det \mathcal{M}_0)^2\Big|^{1/2}.
\end{align*}
But
\begin{align*}
\frac{1}{2}\Tr \mathcal{M}_0^4=\Tr\Big( (S_T^2+D^2)^2+ [S_T,D] [S_T,D]^*\Big)>\Tr D^{4}
\end{align*}
and 
\begin{align*}
\det \mathcal{M}_0=&\det \left(  \begin{array}{cc}D&S_T\\ -S_T&D\end{array}
 \right)=\det D\det\Big(D+S_TD^{-1}S_T \Big)>\det D^2.
\end{align*}
Hence,
 \[1+\frac{1}{2}\Tr \mathcal{M}_0^4+(\det \mathcal{M}_0)^2>1+\Tr D^4+(\det D^2)^2=1+\frac{1}{2}\Tr \mathcal{M}_0^4+(\det \mathcal{M}_0)^2\Big|_{S=0},
\]
and we obtain (\ref{l1.1}).

 The case when  $\mathrm{arg}\,t_1=\mathrm{arg}\,t_2=-\pi/4$  is similar.

Consider now the case $t_1=\tau_1e^{i\pi/4},\, t_2=\tau_2e^{3i\pi/4}$ with $\tau_1>0$ and $\tau_2>0$.  In order  to simplify formulas  below we set
\[S=\left(\begin{array}{cc}s_1&c\\ \bar c& s_2\end{array}\right).
\]
Then straightforward computations yield
\begin{align}\label{l1.0}
\mathcal{D}:= &\det \mathcal{M}(T,S,\lambda I_2)=\det \Big((T+S)(T-S)-\lambda ^2\Big)\\
&=\Big(\tau_1^2\tau_2^2 +2|c|^2\tau_1\tau_2+d(S)\Big)+i\Big(\tau_2^2(s_1^2+\lambda ^2)-\tau_1^2(s_2^2+\lambda ^2)\Big)=:A+iB \notag\\
d(S):=&\det (S^2+\lambda ^2)=|c|^4+2|c|^2(\lambda ^2-s_1s_2)+(s_1^2+\lambda ^2)(s_2^2+\lambda ^2).
\notag\end{align}
Notice first that since $d(S)\ge 0$, if $\lambda =0$, then 
\begin{align*}
|\mathcal D|\ge |\det M(T,0,0)|,
\end{align*}
i.e. (\ref{l1.1}) holds. Thus, it remains to  consider the case $\lambda \ne 0$.

Consider the critical point of $|\mathcal{D}|^2$ with respect to parameters $s_1,s_2,|c|$. Differentiation with respect to
$s_1,s_2$ yields
\begin{align}\label{l1.2}
&\left\{  \begin{array}{l}d'_{s_1}A+2\tau_2^2s_1B=0\\ d'_{s_2}A-2\tau_1^2s_2B=0\end{array}
 \right.\Longrightarrow \tau_1^2s_2d'_{s_1}+\tau_2^2s_1d'_{s_2}=0,
\end{align}
where
\[
d'_{s_1}=2(s_1(s_2^2+\lambda ^2)-|c|^2s_2),\quad d'_{s_2}=2(s_2(s_1^2+\lambda ^2)-|c|^2s_1).
\]
Here we used that $\mathcal{D}\not=0$, hence $A,B$ cannot  be zeros both. The  relations imply
\begin{align}\label{l1.3}
s_1s_2\Big(\tau_1^2(s_2^2+\lambda ^2)+\tau_2^2(s_1^2+\lambda ^2)\Big)-|c|^2\Big(\tau_1^2s_2^2+\tau_2^2s_1^2\Big)=0.
\end{align}
Differentiation with respect to $|c|$  gives
\begin{align}\label{l1.4}
A'_{|c|}=4|c|(\tau_1\tau_2+|c|^2+\lambda ^2-s_1s_2)=0,
\end{align}
since $A>0$. If $c=0$, then  
\begin{align*}
|\mathcal{D}|= |(i\tau_1^2-s_1^2-\lambda ^2)(-i\tau_2^2-s_2^2-\lambda ^2)|
\ge |(i\tau_1^2-\lambda ^2)(-i\tau_2^2-\lambda ^2)|=|\det M(T,0,\lambda I_2)|.
\end{align*}
If $c\not=0$, then, combining (\ref{l1.4}) with (\ref{l1.3}), we get
\begin{align*}
&(\tau_1\tau_2+|c|^2+\lambda ^2)\Big(\tau_1^2(s_2^2+\lambda ^2)+\tau_2^2(s_1^2+\lambda ^2)\Big)-|c|^2\Big(\tau_1^2s_2^2+\tau_2^2s_1^2\Big)=0\\
\Leftrightarrow& |c|^2\lambda ^2(\tau_1^2+\tau_2^2)+(\tau_1\tau_2+\lambda ^2)\Big(\tau_1^2(s_1^2+\lambda ^2)+\tau_2^2(s_2^2+\lambda ^2)\Big)=0.
\end{align*}
But since $\tau_1>0$, $\tau_2>0$, the last relation cannot be valid for $\lambda \ne 0$.

Therefore, inequality (\ref{l1.1}) holds in any critical point of $|\mathcal D|$ with respect to $s_1, s_2$ and $c$.
Notice also, that 
\[
d(S)=\det (S^2+\lambda ^2)\ge \lambda ^2\Tr S^2,
\]
and hence, according to (\ref{l1.0}), 
\begin{align*}
|\mathcal D|^2=A^2+B^2\ge A^2\ge d(S)^2\ge \lambda ^4(\Tr S^2)^2.
\end{align*}
Thus, if $\lambda \ne 0$, $|\mathcal{D}|\to \infty$ if at least one of $|s_1|$, $|s_2|$ or $|c|$ goes to infinity. This implies
that the minimum point of $|\mathcal D|$ with respect to $s_1, s_2$ and $c$ is a finite critical point, and so
(\ref{l1.1}) holds at the minimum point, thus, holds everywhere. 
$\square$

\section{Advanced representation}\label{s:4}

Consider $\mathcal{Z}$ of (\ref{repr_Z}). In this section we make a number of changes of variables 
 transforming (\ref{repr_Z}) to a universal form which allows to prove Theorem \ref{t:1}.
 
 \begin{proposition}\label{p:univ} Given $\mathcal{Z}$ of (\ref{Z}) with $\hat\varepsilon=n^{-1}\hat\epsilon=n^{-1}\diag\{\epsilon_1,\epsilon_2\}$,
 $\hat\varepsilon'=n^{-1}\epsilon' I_2$, we have
 \begin{align}\label{der_Z_fin}
 &\partial_1  \partial_2\frac{\partial}{\partial\bar\zeta_1}\frac{\partial}{\partial\bar\zeta_2}\mathcal Z
 \Big|_{\zeta_1=\zeta_1',\zeta_2=\zeta_2'}=
 \partial_1  \partial_2\frac{\partial}{\partial\bar\zeta_1}\frac{\partial}{\partial\bar\zeta_2}\Phi
 \Big|_{\zeta_1=\zeta_1',\zeta_2=\zeta_2'},
 \end{align}
 where $\partial_1,\partial_2$ are defined by (\ref{pa_1,2}),
 \begin{align}
\Phi &(\hat\zeta,\hat\zeta',\hat\zeta^*,\hat\zeta^{'*},\hat\epsilon,\hat\epsilon')=C\int \exp\Big\{
\mathcal{L}_0(g_2u_*^2, \hat\zeta,\Lambda,U,T,S,F_0 )\Big\}\notag\\
&\times E_* (\hat\epsilon,\hat\epsilon',\chi,\chi^*,R_2,U)d\Lambda dT dSd\kappa d\kappa^* d\chi d\chi^*dU dR_2,
\label{Phi_fin} \\
&\mathcal{L}_0(u_*^2g_2, \hat\zeta,\Lambda,U,T,S,F_0 ):=-\frac{u_*^2g_2}{2} \Sdet\tilde\Delta_0^2+
\rho\,\Sdet F_0\hat\zeta_*F_0^{-1}\hat\zeta-\Tr S^2, \label{L0_fin}\\
&\tilde\Delta_0=F_{0}\ Y_2+Y_1F_{0}^{-1},\label{ti_Delta_0}\\
&F_0=1+P(X^2)+X,\quad  
\label{F_0}\\
&X=\left(\begin{array}{cc}0&\chi\\-\chi^*&0\end{array}\right),\quad 
 Y_{1}=\left(\begin{array}{cc} \Lambda&\kappa \\ \kappa^*&i( T + S) \end{array}\right),\quad 
 Y_{2}=\left(\begin{array}{cc}\Lambda&\kappa \\ \kappa^*&i( T -  S) \end{array}\right),
\label{X,Y} \end{align}
 with  $\Lambda, T, S$, and $R_2>0$ being $2\times 2$ Hermitian matrices, $U$ being a $2\times 2$ unitary matrix, and  $\chi,\chi^*,\kappa,\kappa^*$ being $2\times 2 $ matrices of independent Grassmann, and $\hat\zeta,\hat\zeta_*$ having the form
 \begin{align*}
 &\hat\zeta=\diag\{\zeta,\zeta'_{J_0}\},\quad \hat\zeta_*=\diag\{U\zeta^* U^*,\zeta'^*_{J_0}\},\quad \zeta'_{J_0}=J_0^{-1}\zeta' J_0,\quad
 J_0=R_2^{1/2}(1-\chi^*\chi)^{-1/4}.
 \end{align*}
The function $P(\lambda)$ of (\ref{F_0})  has a form
 \begin{align}\label{P_lam}
P(\lambda)=\sqrt{1+\lambda}-1,
\end{align}
 $\rho$ is defined by (\ref{rho}),  and
\begin{align}\label{E_*}
E_*(\hat\epsilon,\hat\epsilon',\chi,\chi^*,R_2,U)=&\frac{u_*^4 \,\mathrm{Tr}^2\sqrt{1-\chi^*\chi}\cdot \mathrm{Tr}^2\sqrt{1-\chi\chi^*}}{\det^2 R_2}
\exp\Big\{-u_*\Tr ( U\hat\epsilon+\hat\epsilon U^*)\sqrt{1-\chi\chi^*}\\
&-u_*\epsilon'\Tr\big(R_2+R_2^{-1}(1-\chi^*\chi)\big)\Big\}+O(n^{-1/2}).
\notag\end{align}
\end{proposition}
\textit{Proof.}  First
let us transform $\mathcal{L}$ of (\ref{repr_L}). Using  an expression for $\mathcal{Q}$ of (\ref{F_1,F_2}), we get
\begin{align*}
\mathcal{L}(F_1,F_2,\hat\zeta)&=\Str \log \big(I_4\otimes(A_zA_z^*+u_*^2)+(F_1F_2-u_*^2)\otimes I_n\\
&+n^{-1/2}\hat\zeta\otimes A_z^*
+n^{-1/2}F_2^{-1}\hat\zeta_* F_2\otimes A_z+n^{-1}\hat\zeta F_2^{-1}\hat\zeta_* F_2\otimes I_n\big)
-n\,\Str F_1F_2-\Tr S^2.
\end{align*}
where $\hat\zeta$,  $\hat\zeta_*$ are defined in (\ref{hat_zeta}). 
Since
\[\Str \log (I_4\otimes (A_zA_z^*+u_*^2))=0,\]
 using $G$ defined in (\ref{G}), we obtain
\begin{align*}
&\mathcal{L}(F_1,F_2,\hat\zeta)=\Str \log \Big(I_4\otimes I_n+(F_1F_2-u_*^2)\otimes G \\
&+n^{-1/2}\hat\zeta\otimes A_z^*G
+n^{-1/2}F_2^{-1}\hat\zeta_* F_2\otimes A_zG+n^{-1}\hat\zeta F_2^{-1}\hat\zeta_* F_2\otimes G\Big)
-n\,\Str F_1F_2-\Tr S^2.
\end{align*}
Let us introduce new $2\times 2$ Grassmann matrices $\chi,\chi^*,\kappa,\kappa^*$ and $4\times 4$ super-matrices $X$ and $Y_{1,2}$ with the relations (\ref{X,Y})
combined with
\begin{align}\label{change_chi}
&\nu_1=\kappa+n^{1/2}u_*\chi,\quad \nu_2=\kappa-n^{1/2}u_*\chi,\quad
\nu_1^*=\kappa^*+n^{1/2}u_*\chi^*,\quad \nu_2^*=\kappa^*-n^{1/2}u_*\chi^*,\\
&d\nu_1 d\nu_2 d\nu_1^*d\nu_2^*=Cn^{-4} u_*^{-8}d\kappa d\kappa^* d\chi d\chi^*,\notag\\
  & F_{1}=u_*(1+ X)+n^{-1/2}Y_{1},\quad F_{2}=u_*(1- X)+n^{-1/2}Y_2.
\notag\end{align}
Here we used (\ref{Gr_change}) -- (\ref{Gr_shift2}).

Observe that the above  change of  Grassmann variables eliminates the factor $n^4$ in front of integral in (\ref{repr_Z}). Now set
\begin{align}\label{Delta}
\Delta:=&n^{1/2}(F_1F_2-u_*^2)=u_*Y_1+u_*Y_2-n^{1/2}u_*^2X^2+n^{-1/2}Y_1Y_2+u_*XY_2-u_*Y_1X,\\
\Delta_1:=&\hat\zeta\otimes A_z^*G
+F_2^{-1}\hat\zeta_* F_2\otimes A_zG+n^{-1/2}\hat\zeta F_2^{-1}\hat\zeta_* F_2\otimes G. \notag
\end{align}
Then 
\begin{align*}
\mathcal{L}(F_1,F_2,\hat\zeta)=\Str \log (I_4\otimes I_n+n^{-1/2}\Delta\otimes G +n^{-1/2}\Delta_1)
-n^{1/2}\Str \Delta-\Tr S^2.
\end{align*}
Take the function $P$ of (\ref{P_lam})
and  make the change of variables in (\ref{Delta}) 
\[
Y_{l}=Y_{l}'+n^{1/2}u_*P(X^2),\quad l=1,2,
\]
i.e. change (see (\ref{Gr_shift2}))
\begin{align}\label{ch_1}
\Lambda=\Lambda'+n^{1/2}u_*P(-\chi\chi^*),\quad T=T'-in^{1/2}u_*P(-\chi^*\chi).
\end{align}
Then, using  that
\begin{align*}
F_1=u_*F_{0}+n^{-1/2}Y_1',\quad
 F_2=u_*F_{0}^{-1}+n^{-1/2}Y_2'
\end{align*}
with $F_0$ of (\ref{F_0}), we obtain that $\Delta$ takes the form
\begin{align*}
\Delta=&u_*(Y_1'F_0^{-1}+F_0Y_2')+n^{1/2}u_*^2(P^2(X^2)+2P(X^2)-X^2)+n^{-1/2}Y'_1Y'_2.
\end{align*}
Because of (\ref{P_lam}),  the coefficient at $n^{1/2}$ equals to 0. Thus
\begin{align*}
\Delta=u_*(Y_1'F_0^{-1}+F_0Y_2')+n^{-1/2}Y'_1Y'_2.
\end{align*}
Hence
\begin{align}\label{Delta_0}
\Delta=&u_*(Y_1'F_{0}^{-1}+F_{0}Y_2')+O(n^{-1/2}):=u_*\Delta_0+O(n^{-1/2})\\
\mathcal{L}=&n^{1/2}\Str \Delta \Big(\frac{\Tr G}{n} -1\Big)+n^{1/2}\Str \hat\zeta\frac{\Tr A_z^*G}{n}
+n^{1/2}\Str \hat\zeta_* \frac{\Tr A_zG}{n},   \label{L_new}  \\&
+\Str \hat\zeta F_0^{-1}\hat\zeta_* F_0\frac{\Tr G}{n}-\frac{1}{2}\Str \Delta^2 \frac{\Tr G^2}{n}
-\frac{1}{n} \Str (\Delta\otimes G)\Delta_1\notag\\ 
&-\frac{1}{2}\Str (\hat\zeta^2)\frac{\Tr (A_z^*G)^2}{n}
-\frac{1}{2}\Str (\hat\zeta_*)^2 \frac{\Tr (A_zG)^2}{n}\notag\\
&
-\Str \hat\zeta F_0\hat\zeta_* F_0^{-1} \frac{\Tr  A_zGA_z^*G}{n}-\Tr S^2+O(n^{-1/2}).\notag
\end{align}
It is easy to see that
\begin{align*}
\frac{1}{n}\Str (\Delta\otimes G)\Delta_1=\Str \Delta \Big(\hat \zeta \frac{\Tr A^*_zG^2}{n}+F_0\hat\zeta_*F_0^{-1}\frac{\Tr  A_zG^2}{n}\Big)+
O(n^{-1/2}).
\end{align*}
According to (\ref{eq_u}), the coefficient at $\Str \Delta$ in (\ref{L_new})   is zero. Hence, using notations (\ref{rho}),
we get
\begin{align}\label{L.2}
\mathcal{L}=&n^{1/2}\Big(\bar k_A\Str \hat\zeta+ k_A\Str \hat\zeta_*\Big)
-\frac{\bar f_A}{2}\Str \hat\zeta^2-\frac{ f_A}{2}\Str \hat\zeta_*^2\\
&+\Big(1- \frac{\Tr A_zGA_z^*G}{n}\Big)\,\Str \hat\zeta F_{0}\hat\zeta_* F_{0}^{-1}\notag\\
&-\frac{g_2u_*^2}{2}\Str\Delta_0^2
-u_*\Str \Delta_0\Big(\bar h_A\hat \zeta +h_AF_{0}\hat\zeta_*F_{0}^{-1}\Big)-\Tr S^2.
\notag\end{align}
Notice that
\begin{align*}
&-\frac{g_2u_*^2}{2}\Str \Delta_0^2
-u_*\Str \Delta_0\Big(\bar h_A\hat \zeta +h_AF_0\hat\zeta_*F_0^{-1}\Big)
=\frac{\bar h_A^2}{2g_2}\Str \hat\zeta^2+\frac{ h_A^2}{2g_2}\Str (\hat\zeta_*)^2\\
&+\frac{ h_A\bar h_A}{g_2}\Str \hat\zeta F_{0}\hat\zeta_* F_{0}^{-1}-\frac{g_2}{2}\Str \Big(u_*\Delta_0+g_2^{-1}\big(\bar h_A\hat \zeta 
+h_AF_0\hat\zeta_*F_0^{-1}\big)\Big)^2.
\end{align*}
Let us find the matrix $\mathcal{C}$  satisfying the equation
\begin{align}\label{eq_C}
F_{0}\mathcal{C}+\mathcal{C}F_{0}^{-1}=(g_2u_*)^{-1}\Big(\bar h_A\hat \zeta +h_AF_0\hat\zeta_*F_0^{-1}\Big).
\end{align}
The standard method to solve such equations is to consider any matrix $M:\mathbb{C}^2\to\mathbb{C}^2$ as a vector 
in $\mathbb{C}^2\otimes \mathbb{C}^2$
\[M\to M^{(v)}=\sum M_{ij}e_i\otimes e_j.
\]
Then equation (\ref{eq_C}) corresponds to the equation
\begin{align}\label{eq_C1}
F\bar M=(g_2u_*)^{-1}\Big(\bar h_A\hat \zeta +h_A(F_0^{-1}\hat\zeta_*F_0)\Big)^{(v)},\quad F:=F_0\otimes I_4+I_4\otimes (F_0^{-1})^T.
\end{align}
where $(F_0^{-1})^T$ is $F_0^{-1}$-transposed. $F$ is evidently invertible,
since
\[F=2I_4\otimes I_4+\Gamma_1\otimes I_4+I_4\otimes \Gamma_2^T,
\]
where the entries of $\Gamma_1$ and $\Gamma_2$  are polynomials  of Grassmann variables with zero numerical parts.
Therefore, we get 
\[ \mathcal{C}=\frac{1}{u_*g_2}\sum_{m=0}^\infty \frac{(-1)^m(\Gamma_1\otimes I_4+I_4\otimes \Gamma_2^T)^m}{(2u_*)^{m+1}}\Big(\bar h_A\hat \zeta 
+h_AF_0^{-1}\hat\zeta_*F_0\Big).
\]
To check that $\mathcal{C}$ is a super-matrix, it suffices to check that  if $M$ is a super-matrix ,
then the vector $(\Gamma_1\otimes I_4+I_4\otimes \Gamma_2^T)M^{(v)}$ also corresponds to a super-matrix.
But  similarly to (\ref{eq_C1})
\begin{align*}
(\Gamma_1\otimes I_4+I_4\otimes \Gamma_2^T)M^{(v)}\sim \Gamma_1M+M\Gamma_2,
\end{align*}
and the r.h.s. above is a super-matrix since $\Gamma_1 M$ and $ M\Gamma_2$ are super-matrices.

Now, again using (\ref{Gr_shift1}) -- (\ref{Gr_shift2}), we make another change of variables
\begin{align}\label{ch_2}
&Y_{l}'=\tilde Y_{l}-\mathcal{C},\quad  \tilde Y_{l}=\left(\begin{array}{cc}\tilde\Lambda&\tilde\kappa\\ \tilde\kappa^*&\tilde T+(-1)^{l+1} S\end{array}\right), \quad l=1,2;\\
& \tilde\Lambda=\Lambda+\mathcal{C}_{11}, \quad \tilde T=T-i\mathcal{C}_{22},
\quad\tilde\kappa=\kappa+\mathcal{C}_{12}, \quad\tilde\kappa^*=\kappa^*+ \mathcal{C}_{21},
\notag\end{align}
where $\mathcal{C}_{ij}$ is the $ij$-th $2\times 2$-block of the $4\times 4$ super-matrix $\mathcal{C}$.
Then 
\begin{align}\label{ti_Delta}
u_*\Delta_0+g_2^{-1}\big(\bar h_A\hat \zeta +h_AF_0\hat\zeta_*F_0^{-1}\big)\to F_{0}\tilde Y_2+\tilde Y_1F_{0}^{-1}=:u_*\tilde\Delta_0,
\end{align}
and so using notation (\ref{rho}) we can write $\mathcal{L}$ of (\ref{L.2}) in the form
\begin{align*}
\mathcal{L}=&n^{1/2}(\bar k_A\Str \hat\zeta+ k_A\Str \hat\zeta_*)
-\Big(\frac{\bar f_A}{2}-\frac{\bar h_A^2}{2g_2}\Big)\Str \hat\zeta^2-\Big(\frac{f_A}{2}-\frac{ h_A^2}{2g_2}\Big)\Str (\hat\zeta_*)^2\\
&+\Big(1-\frac{\Tr A_zGA_z^*G}{n}+\frac{|h_A|^2}{g_2}\Big)\Str \hat\zeta F_{0}\hat\zeta_* F_{0}^{-1}-\frac{g_2u_*^2}{2}\Str \tilde\Delta_0^2
-\Tr S^2.\notag
\end{align*}
Then identities
\begin{align}\label{G_*.1}
A_z^*G=&G_*A_z^*, \quad \frac{\Tr A_zGA_z^*G}{n}=1-u_*^2\frac{Tr GG_*}{n},
\end{align}
combined with (\ref{eq_u}) yield
\begin{align} \notag
\mathcal{L}=&\mathcal{L}_1+\mathcal{L}_0,\\
\mathcal{L}_1=&n^{1/2}\Big(\bar k_A\Str \hat\zeta+ k_A\Str \hat\zeta_*\Big)
-\Big(\frac{\bar f_A}{2}-\frac{\bar h_A^2}{2g_2}\Big)\Str \hat\zeta^2-\Big(\frac{f_A}{2}-\frac{ h_A^2}{2g_2}\Big)\Str (\hat\zeta_*)^2\label{L_1}
\end{align}
with $\mathcal{L}_0$ of (\ref{L0_fin}).
Hence,
 \begin{align}\label{final}
\mathcal{Z}& (\hat \zeta,\hat \zeta',\hat\epsilon/n,\hat\epsilon'/n)=e^{\mathcal{L}_1}\Phi (\hat\zeta,\hat\zeta',\hat\zeta^*,\hat\zeta^{'*},\hat\epsilon,\hat\epsilon')
 \end{align}
 with $\Phi$ of (\ref{Phi_fin}).
The expression for $E_*$ follows from (\ref{E_*1,2}) and (\ref{repr_Z}), if we make here changes of variables (\ref{ch_1}), (\ref{ch_2})
and $R_2\to u_*R_2$. Indeed, these changes give us 
\begin{align*}
T=\tilde T+i\mathcal{C}_{22}-iu_* \sqrt n P(-\chi^*\chi),
\end{align*}
thus,
\begin{align*}
u_*+in^{-1/2} T=u_* \sqrt{1-\chi^*\chi}+in^{-1/2}(\tilde T+i\mathcal{C}_{22}).
\end{align*}
Similarly,
\begin{align*}
u_*+n^{-1/2}\Lambda= u_* \sqrt{1-\chi\chi^*}+n^{-1/2}(\tilde \Lambda -\mathcal{C}_{11}).
\end{align*}
We also used that
\[
\det \sqrt{1-\chi^*\chi} \cdot \det \sqrt{1-\chi\chi^*}=1.
\]
It follows from the consideration above that we need to consider
\begin{align}\label{der_Z}
 \partial_1  \partial_2\frac{\partial}{\partial\bar\zeta_1}\frac{\partial}{\partial\bar\zeta_2}e^{\mathcal{L}_1}
 \Phi (\hat\zeta,\hat\zeta',\hat\zeta^*,\hat\zeta^{'*},\hat\epsilon,\hat\epsilon')
 \Big|_{\zeta_1=\zeta_1',\zeta_2=\zeta_2'},
 \end{align}
where $\mathcal{L}_1$ is defined in  (\ref{L_1}),    $\Phi$ is defined in (\ref{Phi_fin}), and $\partial_1$, $\partial_2$ are defined in (\ref{pa_1,2}). Since
\begin{align*}
\partial_1\Str \hat\zeta=\partial_2\Str \hat\zeta=0,\quad \partial_1\Str \hat\zeta^2\Big|_{\zeta_1=\zeta_1'}=0,\quad
\partial_2\Str \hat\zeta^2\Big|_{\zeta_2=\zeta_2'}=0,
\end{align*}
we obtain zero in (\ref{der_Z})  if at least one of the derivatives $\partial_1 , \partial_2$ is applied to $e^{\mathcal{L}_1}$.
Moreover, since for any $\xi\in\mathbb{C}$
\begin{align*}
&\mathcal{Z}( \zeta_1+\xi,\zeta_2,\zeta_1'+\xi,\zeta_2',\bar\zeta_1+\bar\xi,\bar\zeta_2,\bar\zeta_1'+\bar\xi,\bar\zeta_2')
\Big|_{\zeta_1=\zeta_1',\bar\zeta_1=\bar\zeta_1'},\\
&\exp\{\mathcal{L}_1( \zeta_1+\xi,\zeta_2,\zeta_1'+\xi,\zeta_2',\bar\zeta_1+\bar\xi,\bar\zeta_2,\bar\zeta_1'+\bar\xi,\bar\zeta_2')\}
\Big|_{\zeta_1=\zeta_1',\bar\zeta_1=\bar\zeta_1'}
\end{align*}
do not depend on $\xi$, we observe that $\Phi$ possesses   the same property. Then differentiating with respect to $\Re\xi$ and $\Im\xi$
we obtain
\begin{align*}
 \partial_1\Phi (\hat\zeta,\hat\zeta',\hat\zeta^*,\hat\zeta^{'*},\hat\epsilon,\hat\epsilon')
 \Big|_{\zeta_1=\zeta_1',\bar\zeta_1=\bar\zeta_1'}=0,
\end{align*}
and similarly
\[\partial_2\Phi (\hat\zeta,\hat\zeta',\hat\zeta^*,\hat\zeta^{'*},\hat\epsilon,\hat\epsilon')
 \Big|_{\zeta_2=\zeta_2',\bar\zeta_2=\bar\zeta_2'}=0.\]
Thus we conclude that only a term in (\ref{der_Z}) which has all derivatives applied to $\Phi$ gives a non-zero contribution which implies
(\ref{der_Z_fin}).

$\square$

\section{Proof of the main results}

\subsection{Proof of Proposition \ref{p:logZ}}\label{s:5}

Denote $\mathcal{D}(\epsilon_1,\epsilon_2,)$ the l.h.s. of (\ref{p.logZ}).
We use a trivial formula
\begin{align*}
\mathcal{D}(\epsilon_1,\epsilon_2)=&\int_0^{\epsilon_1}\int_0^{\epsilon_2}\frac{\partial^2}{\partial\epsilon_1\partial\epsilon_2}
\mathbb{E}\{\log( Y(z_1)+(\epsilon_1/n)^2)\log ( Y(z_2)+(\epsilon_2/n)^2)\}d\epsilon_1d\epsilon_2\\
=&\int_0^{\epsilon_1}\int_0^{\epsilon_2}\frac{\partial^2}{\partial\epsilon_1\partial\epsilon_2}
\mathcal{Z}(\zeta,\zeta',\hat\epsilon/n, \hat\epsilon'/n)\Big|_{\epsilon_1'=\epsilon_1,\epsilon_2'=\epsilon_2}d\epsilon_1d\epsilon_2.
\end{align*}
Observe that if we denote $\mathcal{G}(z,\epsilon)=( Y(z)+(\epsilon/n)^2)^{-1}$, then
\begin{align*}
&\frac{\partial^2}{\partial\epsilon_1\partial\epsilon_2}
\mathcal{Z}(\zeta,\zeta',\hat\epsilon/n, \hat\epsilon'/n)\Big|_{\epsilon_1'=\epsilon_1,\epsilon_2'=\epsilon_2}=
\frac{\epsilon_1\epsilon_2}{n^4}E\{\Tr \mathcal{G}(z_1,\epsilon_1)\Tr\mathcal{G}(z_2,\epsilon_2)\}\\
&\le \frac{\epsilon_1\epsilon_2}{n^4}E^{1/2}\{(\Tr \mathcal{G}(z_1,\epsilon_1))^2\}E^{1/2}\{(\Tr(\mathcal{G}(z_2,\epsilon_2))^2\},
\end{align*}
and
\begin{align*}
\frac{\epsilon^2}{n^4}E\{(\Tr \mathcal{G}(z,\epsilon))^2\}=\frac{\partial^2}{\partial\epsilon_1\partial\epsilon_2}
\mathcal{Z}(0,0,\hat\epsilon/n, \epsilon I_2/n)\Big|_{\hat\epsilon=\epsilon I_2}.
\end{align*}
Using (\ref{final})  for $\mathcal{Z}(0,\hat\epsilon/n, \epsilon'I_2/n)$,
one can see that the matrix $\epsilon$ appears only in the term $E_{*1}(\hat\epsilon,\Lambda,U)$  and
\begin{align}
\frac{\partial^2}{\partial\epsilon_1\partial\epsilon_2}E_{*1}(\hat\epsilon,\Lambda,U)=u_*^2
&e^{-u_*\Tr (\Lambda U\hat\epsilon+\hat\epsilon(\Lambda U)^*)}((U\Lambda )_{11}+(\Lambda U)^*_{11})
((\Lambda U)_{22}+(\Lambda U)^*_{22}).
\label{ti-E_*}\end{align}
To integrate with respect to $U$ observe that for $\zeta=\zeta'=0$ $U$ appears in  (\ref{final}) only in $E_{*1}$.
Changing variable $U\to UD$ with $D=\mathrm{diag}\{e^{i\phi_1},e^{i\phi_2}\}$
and integrating with respect to $\phi_1,\phi_2$, we obtain
\begin{align*}
\mathcal{I}_U=&\int dU \frac{\partial^2E_{*1}(\hat\epsilon,\Lambda,U)}{\partial\epsilon_1\partial\epsilon_2}\\
=&u_*^2(2\pi)^{-2}\int dU\int_0^{2\pi} d\phi_1((e^{i\phi_1}(\Lambda U)_{11}+e^{-i\phi_1}(\Lambda U)^*_{11})e^{\epsilon_1(e^{i\phi_1}((\Lambda U)_{11}+e^{-i\phi_1}(\Lambda U)^*_{11})}\\
&\times\int_0^{2\pi} d\phi_2 (e^{i\phi_2}((\Lambda U )_{22}+e^{-i\phi_2}(\Lambda U)^*_{22})
e^{\epsilon_2(e^{i\phi_2}((\Lambda U)_{22}+e^{-i\phi_2}(\Lambda U)^*_{22})}.
\end{align*}
Then, integrating  by parts with respect $\phi_1,\phi_2$, we conclude that 
\[
\mathcal{I}_U=\epsilon_1\epsilon_2 \tilde I_U
\]
with some $\tilde I_U$ uniformly  bounded in $\epsilon_1,\epsilon_2$. Now the inequality
\begin{align*}
E_{*2}(\epsilon')=\int_0^\infty\int_0^\infty(\rho_1^{-1}-\rho_2^{-1})^2e^{-\epsilon'(\rho_1+\rho_2+u_*^2(\rho_1^{-1}+\rho_2^{-1}))}d\rho_1d\rho_2
\le C(\epsilon')^{-2}
\end{align*}
yields the bound
\[
\Big|\frac{\partial^2}{\partial\epsilon_1\partial\epsilon_2}
\mathcal{Z}(0,0,\hat\epsilon/n, \epsilon I_2/n)\Big|_{\epsilon_1=\epsilon,\epsilon_2=\epsilon}\Big|\le C,
\]
which completes the proof of Proposition \ref{p:logZ}. $\square$

\subsection{Proof of Theorem \ref{t:1}}\label{s:4.1}
To prove Theorem \ref{t:1}, take any finitely supported function $\varphi(\zeta_1,\bar\zeta_1,\zeta_2,\bar\zeta_2)$ which possesses 4 bounded 
derivatives  and  write, using (\ref{der_Z_fin}) combined with integration by parts with respect to $d\zeta_1d\bar\zeta_1d\zeta_2d\bar\zeta_2$,
\begin{align*}
&\int \varphi(\zeta_1,\bar\zeta_1,\zeta_2,\bar\zeta_2)
\Big(    \partial_1  \partial_2\frac{\partial}{\partial\bar\zeta_1}\frac{\partial}{\partial\bar\zeta_2}
 \Phi (\hat\zeta,\hat\zeta',\hat\zeta^*,\hat\zeta^{'*},\hat\epsilon,\hat\epsilon')
 \Big|_{\zeta_1=\zeta_1',\bar\zeta_1=\bar\zeta_1'}\\
 &\hskip4cm- \frac{\partial^2}{\partial\zeta_1\partial\bar\zeta_1} \frac{\partial^2}{\partial\zeta_2\partial\bar\zeta_2} \E\{\log \det Y(z_1)\log \det Y(z_2)\} 
 \Big)d\zeta_1d\bar\zeta_1d\zeta_2d\bar\zeta_2\\
 =&\int \frac{\partial^2}{\partial\zeta_1\partial\bar\zeta_1} \frac{\partial^2}{\partial\zeta_2\partial\bar\zeta_2} \varphi(\zeta_1,\bar\zeta_1,\zeta_2,\bar\zeta_2)
 \Big(\mathbb{E}\Big\{\log\det\Big( Y(z_1)+\Big(\frac{\epsilon_1}{n}\Big)^2\Big)
\log\det\Big ( Y(z_2)+\Big(\frac{\epsilon_2}{n}\Big)^2\Big)\Big\}\\
&\hskip4cm-\mathbb{E}\{\log \det Y(z_1)\log \det Y(z_2)\}\Big)d\zeta_1d\bar\zeta_1d\zeta_2d\bar\zeta_2.
\end{align*}
Since according to Proposition \ref{p:logZ} the r.h.s. here tends to 0, as $\epsilon\to 0$, we conclude that in the weak sense
\begin{align*}
&\lim_{n\to\infty} \rho_{n}(z_0+\zeta_1n^{-1/2},z_0+\zeta_2n^{-1/2})\\&\hskip2cm =
(4\pi)^{-2}\lim_{\epsilon\to 0}\lim_{n\to\infty} \partial_1  \partial_2\frac{\partial}{\partial\bar\zeta_1}\frac{\partial}{\partial\bar\zeta_2}
\Phi (\hat\zeta,\hat\zeta',\hat\zeta^*,\hat\zeta^{'*},\hat\epsilon,\hat\epsilon')\Big|_{\zeta_1=\zeta_1',\bar\zeta_1=\bar\zeta_1',\hat\epsilon=\hat\epsilon'=
 \epsilon I_2}.
\end{align*}
On the other hand,  in the case  $A_0=0$ (i.e. pure Ginibre case), $|\tilde z_0|^2=1-g_2u_*^2$ (recall that $g_2u_*^2<n^{-1}\Tr G=1$), 
$\tilde\zeta=\rho^{1/2}\zeta$, $\tilde\zeta'=\rho^{1/2}\zeta'$ we obtain the same expression (\ref{final}) for $\Phi$. Hence 
the limit which we get after differentiation in the r.h.s. of (\ref{der_Z_fin}) and then sending $\epsilon\to 0$
up to the multiplicative constant coincides with that for $A_0=0$ with parameters $\tilde z_0,\tilde\zeta$ chosen above:
\begin{align*}
&\lim_{n\to\infty}p_2(z+\zeta_1/\sqrt n,z+\zeta_2/\sqrt n)=
C_0(1-e^{-\rho|\zeta_1-\zeta_2|^2}),
\end{align*}
which coincides with (\ref{t1.1}). Here we also used Proposition \ref{p:logZ}.

The constant $C_0$ can be found from the condition
\[
\lim_{|\zeta_1-\zeta_2|\to\infty} \lim_{n\to\infty}p_2(z+\zeta_1/\sqrt n,z+\zeta_2/\sqrt n)=\rho^2
\]
which concludes the proof of Theorem \ref{t:1}.

\section{Appendix}

\begin{lemma}\label{A = B^2}
For any integrable function $f$ we have
\begin{equation} \label{changeLem}
\int\limits_{\herm_2^+}  f(A) \,dA= 4 \int\limits_{\herm_2^+} (\Tr B)^2 \det B \cdot  f(B^2)\,dB.
\end{equation}
\end{lemma}
\begin{proof}
Let us change the variables $A = V^*MV$, where $V$ is a unitary matrix and $M = \diag \lbrace \mu_1, \mu_2 \rbrace$, $\mu_1, \mu_2 > 0$. The Jacobian is $2\pi(\mu_1 - \mu_2)^2$ (see e.g.\ \cite{Hu:63}). We get
\begin{equation*}
\int\limits_{\herm_2^+}  f(A) dA= 2\pi \int (\mu_1 - \mu_2)^2  f(V^*MV)\, dM dV.
\end{equation*}
Then change the variables $\mu_j \to \mu_j^2$, $j = 1,2$. It follows
\begin{equation*}
\int\limits_{\herm_2^+}  f(A) dA= 8\pi \int (\mu_1 + \mu_2)^2(\mu_1 - \mu_2)^2\mu_1\mu_2 f(V^*M^2V)\,dM dV.
\end{equation*}
Finally, we do the reverse change $B = V^*MV$ and get \eqref{changeLem}.
\end{proof}
\begin{lemma}\label{W = Lambda U}
For any integrable function $f$ we have
\begin{equation} \label{changeLem1}
\int  f(W) dW^*dW= 2\pi^3 \int\limits_{\herm_2^+} (\Tr \Lambda)^2 \det \Lambda d\Lambda \int\limits_{U(2)} dU\, f(\Lambda U),
\end{equation}
where the first integral is over the space of $2 \times 2$ matrices with complex entries.
\end{lemma}
\begin{proof}
Let us change the variables $W = V^*MU$, where $U$ and $V$ are unitary matrices and $M = \diag \lbrace \mu_1, \mu_2 \rbrace$, $\mu_1, \mu_2 > 0$. The Jacobian is $4\pi^4(\mu_1^2 - \mu_2^2)^2 \mu_1\mu_2$ (see e.g.\ \cite{Hu:63}). We get
\begin{equation*}
\int dW^*dW\, f(W) = 4\pi^4 \int (\mu_1^2 - \mu_2^2)^2 \mu_1\mu_2  f(V^*MU) \,dM dU dV.
\end{equation*}
Then change the variables $\Lambda = V^*MV$ with a Jacobian $\left(2\pi(\mu_1 - \mu_2)^2\right)^{-1}$ (see e.g.\ \cite{Hu:63}). It follows
\begin{equation*}
\int dW^*dW\, f(W) = 2\pi^3 \int (\mu_1 + \mu_2)^2\mu_1\mu_2 d\Lambda dU f(\Lambda V^*U).
\end{equation*}
A Haar measure is invariant w.r.t. shifts. Therefore, after the shift $U \to VU$ we immediately obtain \eqref{changeLem1}.
\end{proof}

{\bf Statements and Declarations.} 

\textbf{Funding:}  M.S. was supported in part by  Deutsche Forschungsgemeinschaft (DFG) grant SFB 1283/2 2021 -- 317210226 and
 the Robbert Dijkgraaf Member Fund at the Institute for Advanced Study.
T.S. was supported in part by Alfred P. Sloan Foundation grant FG-2022-18916. 

\textbf{Data Availability:} Data sharing is not applicable to this article as no datasets were generated or analysed during the current study.

\end{document}